\documentclass[journal]{IEEEtran}

\usepackage{graphicx}
\usepackage{cite}
\usepackage{bm}
\usepackage{amsmath}
\usepackage{amsfonts}
\usepackage{amssymb}
\usepackage{amsthm}
\usepackage{stmaryrd}
\allowdisplaybreaks

\usepackage{color}
\usepackage{ulem}
\usepackage{stfloats}
\usepackage{float}

\usepackage[font=footnotesize,labelsep=period]{caption}
\usepackage[caption=false,font=footnotesize]{subfig}

\usepackage{multirow}
\usepackage{booktabs}
\usepackage{makecell}

\usepackage{algorithm}
\usepackage{algpseudocode}
\makeatletter

\newtheorem{prop}{Proposition}
\newtheorem{lemma}{Lemma}

\usepackage{hyperref} % 关闭彩色文字，用方框代替
\newtheorem{remark}{Remark}

\begin{document}
	
\title{\LARGE Toward Alias-Free Channel Extrapolation in Upper Mid-Band Systems: \\A Spatial-Frequency-Temporal Tensor Learning Approach
}

\author{Jiawei~Zhuang, 
	Hongwei Hou,
    Yafei Wang,~\IEEEmembership{Graduate Student Member,~IEEE}, Xinping Yi,~\IEEEmembership{Member,~IEEE}, Wenjin~Wang,~\IEEEmembership{Member,~IEEE}, Jiangzhou Wang,~\IEEEmembership{Fellow,~IEEE}, Bj{\"o}rn Ottersten,~\IEEEmembership{Fellow,~IEEE}
    \thanks{\textit{(Jiawei Zhuang and Hongwei Hou contributed equally to this work.) (Corresponding
    author: Wenjin Wang.)}}
	\thanks{
		 Jiawei Zhuang, Hongwei Hou, Yafei Wang, Wenjin Wang and Jiangzhou Wang are with the National Mobile Communications Research
		Laboratory, Southeast University, Nanjing 210096, China, and also
		with Purple Mountain Laboratories, Nanjing 211100, China (e-mail:
        \{jwzhuang, hongweihou, wangyf, wangwj, j.z.wang\}@seu.edu.cn).} 
        	\thanks{
		Xinping Yi is with the National Mobile Communications
Research Laboratory, Southeast University, Nanjing 210096, China (e-mail:
xyi@seu.edu.cn).} 
\thanks{Bj{\"o}rn Ottersten is with the Interdisciplinary Centre for Security, Reliability and Trust (SnT), University of Luxembourg, Luxembourg (e-mail: bjorn.ottersten@uni.lu).}
			
}

% \markboth{Journal of \LaTeX\ Class Files,~Vol.~14, No.~8, August~2015}%
% {Shell \MakeLowercase{\textit{et al.}}: Bare Demo of IEEEtran.cls for IEEE Journals}

\maketitle
\begin{abstract}
Upper mid-band massive multiple-input multiple-output (MIMO) offers a favorable capacity-coverage trade-off for next-generation wireless systems, but its large antenna arrays, wide bandwidths, and faster temporal variation substantially increase the pilot overhead required for accurate channel state information (CSI) acquisition.
   To reduce this overhead, this paper establishes a tensor-structured multi-domain channel extrapolation framework that exploits the limited-scattering nature of practical propagation environments to recover complete CSI across the spatial-frequency-temporal (SFT) domains from limited observations.
   Specifically, we develop a Tucker-based SFT-domain signal model to represent the complete CSI, where the factor matrices are parameterized by angle-delay-Doppler (ADD)-domain grids.
   Thanks to this representation, we reveal that limited SFT-domain observations imposed by uniform pilot patterns and antenna-port selection inherently induce ADD-domain aliasing, so that multiple physically distinct ADD-domain components become indistinguishable within structured ADD aliasing groups.
  To tackle this issue, we introduce a support-prior-assisted ADD-domain de-aliasing mechanism that leverages coarse-grained support information. Since exact closed-form characterization of this mechanism is difficult to derive, we propose a tensor-structure-aware axial-attention neural network (TANN), which integrates axis-wise attention with a lightweight multi-scale CNN-based gating module to incorporate support priors for ADD-domain de-aliasing.
With tensor-structure modeling and mixed-configuration training over different
pilot decimation factors, TANN yields
a unified model that generalizes across pilot configurations without retraining.
Extensive numerical results demonstrate that the proposed framework consistently outperforms representative baselines across diverse pilot configurations, user velocities, carrier frequencies, and propagation scenarios.
   \end{abstract}


% \fi
\IEEEpeerreviewmaketitle

\section{Introduction}
\label{sec_intro}

\IEEEPARstart{D}{riven} by the growing demand for data-intensive 
applications such as broadband Internet of Things (IoT)~\cite{BroadbandIoTZhou2023} and extended 
reality (XR)~\cite{XRLoopbackMechanismBojovic2023}, next-generation wireless networks are envisioned 
to deliver high capacity and broad coverage simultaneously. 
However, the currently deployed frequency ranges (FRs), namely FR1
(sub-6~GHz) and FR2 (millimeter-wave), remain unable to simultaneously
achieve high throughput under limited
bandwidth availability and provide wide-area coverage ~\cite{FR1FR2ChannelMeasurementNaqvi2021}.
In this context, growing attention has been directed toward the newly
defined  FR3 (upper mid-band, 7--24~GHz), which offers a favorable
trade-off between  spectrum exhaustion in FR1 and 
severe attenuation in FR2~\cite{XLMIMOCPHou}.

To realize the envisioned capacity and coverage of FR3, accurate channel state information (CSI) acquisition is essential for enabling effective wireless transmission designs, including user scheduling, precoding, and signal detection~\cite{11049893}.
CSI is conventionally acquired via pilot-based channel estimation. 
However, the large antenna arrays, wide bandwidths, and elevated carrier
frequencies envisioned for FR3 systems collectively impose prohibitive pilot
overhead across the spatial, frequency, and temporal domains, thereby posing
fundamental challenges to accurate CSI acquisition:
\begin{itemize}
    \item \textit{Spatial domain:} Hybrid beamforming (HBF) is widely adopted in wideband massive multiple-input multiple-output (MIMO) systems to reduce hardware cost and power consumption by limiting the number of radio frequency (RF) chains at the base station (BS)~\cite{HBFSurveyAhmed2018}. Nevertheless, acquiring full-dimensional CSI under HBF requires multi-slot pilot training with different analog beamforming matrices, incurring substantial pilot overhead and considerable time and power consumption~\cite{HBFSwithchChannelEstimationMao2018}.
    \item \textit{Frequency domain:} The wide bandwidth of FR3 leads to larger subcarrier spacing in orthogonal frequency-division multiplexing (OFDM) systems\cite{OFDM1}\cite{OFDM2}, which induces rapid channel variations across subcarriers and exacerbates frequency-selective fading. To ensure reliable channel estimation, frequency-domain pilots must be placed more densely, thereby increasing pilot overhead.
    \item \textit{Temporal domain:} 
    The elevated carrier frequency of FR3 results in larger Doppler shifts and faster channel variations than FR1, thereby exacerbating channel aging under mobility. Consequently, a shorter pilot transmission interval is required to track the channel variation, substantially increasing temporal-domain pilot overhead~\cite{FR3ChannelMeasurementMiao2023}.
\end{itemize}
Therefore, multi-domain channel extrapolation has emerged as a promising solution to reduce pilot overhead by reconstructing and predicting CSI across antennas, subcarriers, and OFDM symbols from limited observations.
\vspace{-2.5mm}
\subsection{Prior Work}

Due to the limited-scattering nature of practical propagation environments \cite{3GPP_38901}, wireless channels are often governed by a small number of dominant propagation paths, giving rise to structured correlations across the spatial, frequency, and temporal domains. Accordingly, existing channel extrapolation methods exploit such correlations either within individual domains or jointly across multiple domains, and can be broadly categorized into spatial-, frequency-, and temporal-domain approaches, as well as multi-domain combinations.
For the spatial domain, prior works demonstrated that full-dimensional CSI can be reconstructed from observations at only a subset of antenna ports \cite{SpatialDomainChannelExtrapolationZhang2022}.
Since the extrapolation performance strongly depends on the antenna-port selection pattern at the BS, recent studies have developed learning-based frameworks that jointly optimize antenna selection and channel extrapolation~\cite{yang2020deep}\cite{Lin2021AntennaSelection}, with the effectiveness of this joint design paradigm further validated in massive MIMO systems on a 3rd generation partnership project (3GPP)-compliant over-the-air prototype~\cite{GaoFeiFei2025AntennaSelection}.
Beyond optimizing the antenna observation pattern, recent studies have exploited auxiliary information to improve spatial-domain extrapolation accuracy. Specifically, environment-related multipath information extracted from CSI was incorporated into the extrapolation network through cross-attention in \cite{Gao2026PDP}. In a complementary direction, uplink and downlink reference-signal observations were fused for full CSI reconstruction in the spatial domain~\cite{He2025JUDCEN}.

In the frequency domain, channel extrapolation aims to recover full-band CSI from limited pilot subcarriers to reduce pilot overhead.
A straightforward approach is to perform least-squares (LS) estimation on pilot subcarriers and then interpolate the resulting CSI to data subcarriers~\cite{OFDMLinearInterpolationLee2008}. To further improve estimation accuracy, linear minimum 
mean squared error (LMMSE) estimators were developed in \cite{MMSEOFDMChannelEstimationEdfors1998}, incorporating channel covariance information to better exploit frequency-domain channel correlations.
Nevertheless, such model-driven methods are limited by fixed interpolation rules or statistical assumptions, and their performance may degrade under low-density pilot configurations or mismatched channel statistics. 
This has motivated learning-based approaches to capture richer channel structures from data. For example, a dual-attention-based channel estimation network was proposed in~\cite{GaofeifeiDACEN2023} to exploit spatial-delay features for frequency-domain channel extrapolation under low-density pilot configurations.

%信道预测
In the temporal domain, channel extrapolation predicts future CSI based on historical observations, thereby avoiding frequent transmission of pilots.
Model-based methods, including the sum-of-sinusoids model \cite{SosChannelPredictionWong2006}, autoregressive models \cite{WuchiARChannelPrediction,Yuan2020MachineLearningBasedChannelPredictionAR} and parametric channel models \cite{ParameterBasedModel2018prediction}, have shown promising performance in channel prediction. The advancement of deep learning has further enhanced channel prediction performance by enabling data-driven modeling of complex temporal dynamics.
Prior data-driven approaches primarily relied on recurrent neural networks (RNNs) to capture temporal correlations in CSI sequences, achieving improved prediction accuracy over classical approaches~\cite{LSTMGRUChannelPredictionHelmy2023}.
To alleviate the error accumulation caused by serial processing, Transformer-based channel prediction was introduced in~\cite{TrasnformerChannelPredictionJiang2022}, leveraging self-attention to capture long-term temporal dependencies.

Building upon the progress of single-domain channel extrapolation, recent studies have explored multi-domain channel extrapolation to 
improve CSI acquisition efficiency by  exploiting coupled structures across multiple domains.
As an early effort, a two-stage two-dimensional (2D)  scheme was proposed in \cite{2DChannelExtrapolationWan2024} to exploit the frequency-temporal channel structure, thereby improving extrapolation robustness against practical impairments. Beyond the frequency-temporal domain, a unified multi-domain channel extrapolation framework was developed in \cite{MultiDomainChannelExtrapolationHan2021}, where the time, frequency, polarization, spatial, and user domains were  considered with sparse pilot pattern design for efficient channel extrapolation. Since the cross-domain coupling inherent in practical wireless channels arises from complex scattering environments and is difficult to characterize analytically, learning-based approaches have been introduced for multi-domain channel extrapolation.
In~\cite{3DDomainExtrapolationGaoFeiFei}, a unified three-dimensional (3D) framework based on deep learning was developed for spatial-frequency-temporal (SFT)-domain channel extrapolation, enabling full-dimensional uplink CSI recovery from sparse spatial-frequency pilots and slot-level downlink CSI prediction by exploiting uplink-downlink reciprocity and temporal correlations.
To reduce the computational burden of joint cross-domain modeling, a Transformer encoder-like architecture was developed in~\cite{Enabling6GThroughMultiDomainChannelExtrapolation} for joint time-frequency-space extrapolation, using a multilayer perceptron (MLP) in place of multi-head attention.

\vspace{-4mm}
\subsection{Motivation and Main Contributions}
The aforementioned single-domain channel extrapolation approaches cannot jointly exploit the coupled correlations across the SFT domains, which limits their capability to reduce pilot overhead and enhance extrapolation accuracy. Although several preliminary attempts mitigate this issue by leveraging multi-domain correlations, they mainly operate in the SFT domains, without explicitly exploiting the sparse structures in angle-delay domains enabled by large-scale antenna arrays and wide bandwidths, nor the Doppler-domain structure that better characterizes temporal channel variations. Moreover, model-driven methods often depend on restrictive assumptions and are vulnerable to environmental mismatch, whereas data-driven methods are usually tailored to specific pilot configurations and may require retraining when the pilot pattern changes. 
To this end, we investigate multi-domain channel extrapolation for upper mid-band massive MIMO-OFDM systems, leveraging an angle-delay-Doppler (ADD) representation to capture limited-scattering channel structures and adopting a unified model to
generalize across diverse pilot configurations. The main contributions of this work are summarized as follows:
\begin{itemize}
    \item
    By exploiting the limited-scattering propagation characteristics of wireless channels, we develop a tensor-structured SFT-domain channel model along with its ADD-domain representation, where the factor matrices are parameterized by the discrete grids of the corresponding domains. Building upon this model, we analytically reveal the limitation of uniform pilot decimation in the SFT domains that induces asymmetric aliasing across ADD domains, in such a way that the aliasing ADD-domain components are rendered indistinguishable from decimated observations. To resolve this aliasing ambiguity, we reformulate multi-domain channel extrapolation as a support-prior-assisted ADD-domain de-aliasing problem, where support priors obtainable from historical observations, environmental context, or auxiliary reference signals are incorporated to guide ambiguity resolution.

    \item 
    To tackle the support-prior-assisted ADD-domain de-aliasing problem, we propose a tensor-structure-aware axial-attention neural network (TANN), which integrates an axial-attention backbone with a support-prior processing gating module. The axial-attention backbone exploits the tensor-structured ADD-domain representation to decompose full self-attention over the joint ADD-domain into sequential axis-wise self-attention along the angle, delay, and Doppler domains, alleviating the computational bottleneck of Transformer-based modeling under large antenna arrays and wide bandwidths. 
    To effectively incorporate coarse domain-wise support priors,
    we design a lightweight multi-scale convolutional neural network (CNN) gating module that extracts local support features and neighborhood contiguity, and adaptively modulates axial-attention features to guide ADD-domain de-aliasing. Moreover, a learnable per-axis residual gate is incorporated to adaptively regulate prior-guided feature updates, enhancing robustness to imperfect priors.

    \item 
    To incorporate the ADD-domain sparsity prior into the training process, we  design a composite loss that combines the SFT-domain normalized mean squared error (NMSE) with auxiliary axis-wise ADD-domain spectral consistency terms.
The training samples are generated under 3GPP-defined outdoor propagation scenarios, and the corresponding support priors are provided as auxiliary inputs for learning a prior-conditioned de-aliasing operator. 
During deployment, the support-prior input can be derived from the specific propagation conditions observed at the target site, allowing the model to adapt to site-dependent propagation characteristics without site-wise retraining.
To enable pilot-configuration-generalizable extrapolation, we develop a mixed-configuration training strategy that jointly trains the model over randomized pilot decimation factors. By mapping different pilot configurations into a common ADD-domain input space through flexible transform sampling densities, the proposed framework learns a unified model that generalizes across diverse pilot configurations  without retraining.

\end{itemize}

\vspace{-1mm}
This paper is structured as follows: Section~\ref{sec_system} presents the system and signal models. Section~\ref{sec_problem} formulates the problem through multi-domain aliasing analysis and a support-prior-assisted de-aliasing framework. Section~\ref{sec_network} details the proposed network architecture, together with the loss design and training strategy.
Numerical results are given in 
Section~\ref{sec_simulation}. Finally, Section~\ref{sec_Conclusion} concludes this paper. 

{\it{Notations}}: Scalars, vectors, matrices, and tensors are denoted by $x$, $\mathbf{x}$, $\mathbf{X}$, and $\boldsymbol{\mathcal{X}}$, respectively.  The superscripts $(\cdot)^*$, $(\cdot)^T$, $(\cdot)^H$, and $(\cdot)^{\dagger}$ denote conjugate, transpose, conjugate transpose, and pseudo-inverse, respectively. The norms $\|\cdot\|_{\mathrm{F}}$ and $\|\cdot\|_2$ denote the Frobenius and Euclidean norms, respectively. 
$\odot$ and $\circ$ denote the Hadamard and outer products, respectively, and $\times_d$ denotes the mode-$d$ tensor-matrix product. 
$\mathcal{CN}(\mu,\sigma^2)$ denotes the complex Gaussian distribution with mean $\mu$ and variance $\sigma^2$, and $\mathbb{E}\{\cdot\}$ denotes expectation.

\section{System Model}
\label{sec_system}
\begin{figure}[!t]
	\centering
	%	\hspace*{0.05\linewidth}
	\subfloat[Spatial-domain channel extrapolation under the HBF architecture.]{\includegraphics[width = 0.935\linewidth]{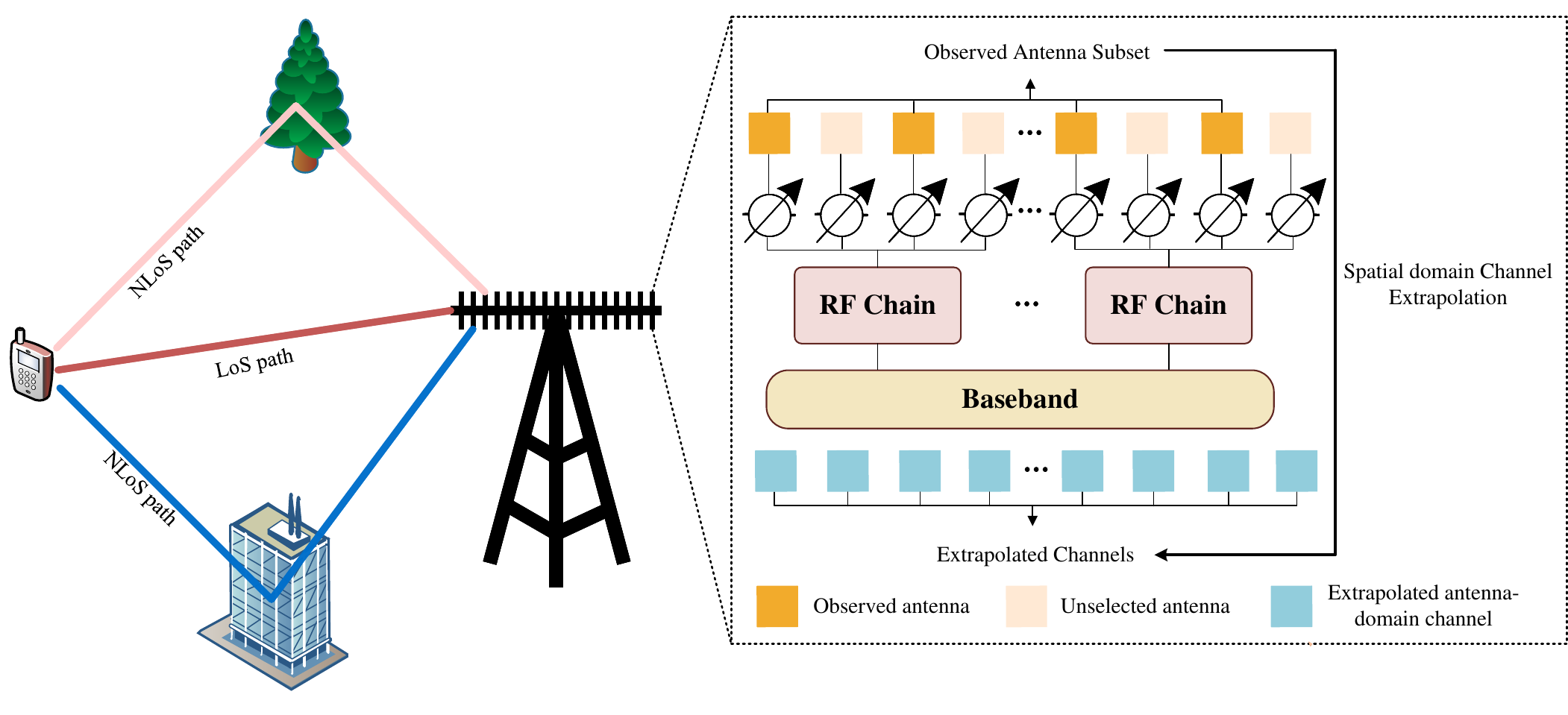}%
	\label{fig:AntennaDomainChannelExtrapolation}}
    
	\subfloat[Frequency- and temporal-domain channel extrapolation.]{\includegraphics[width = 1\linewidth]{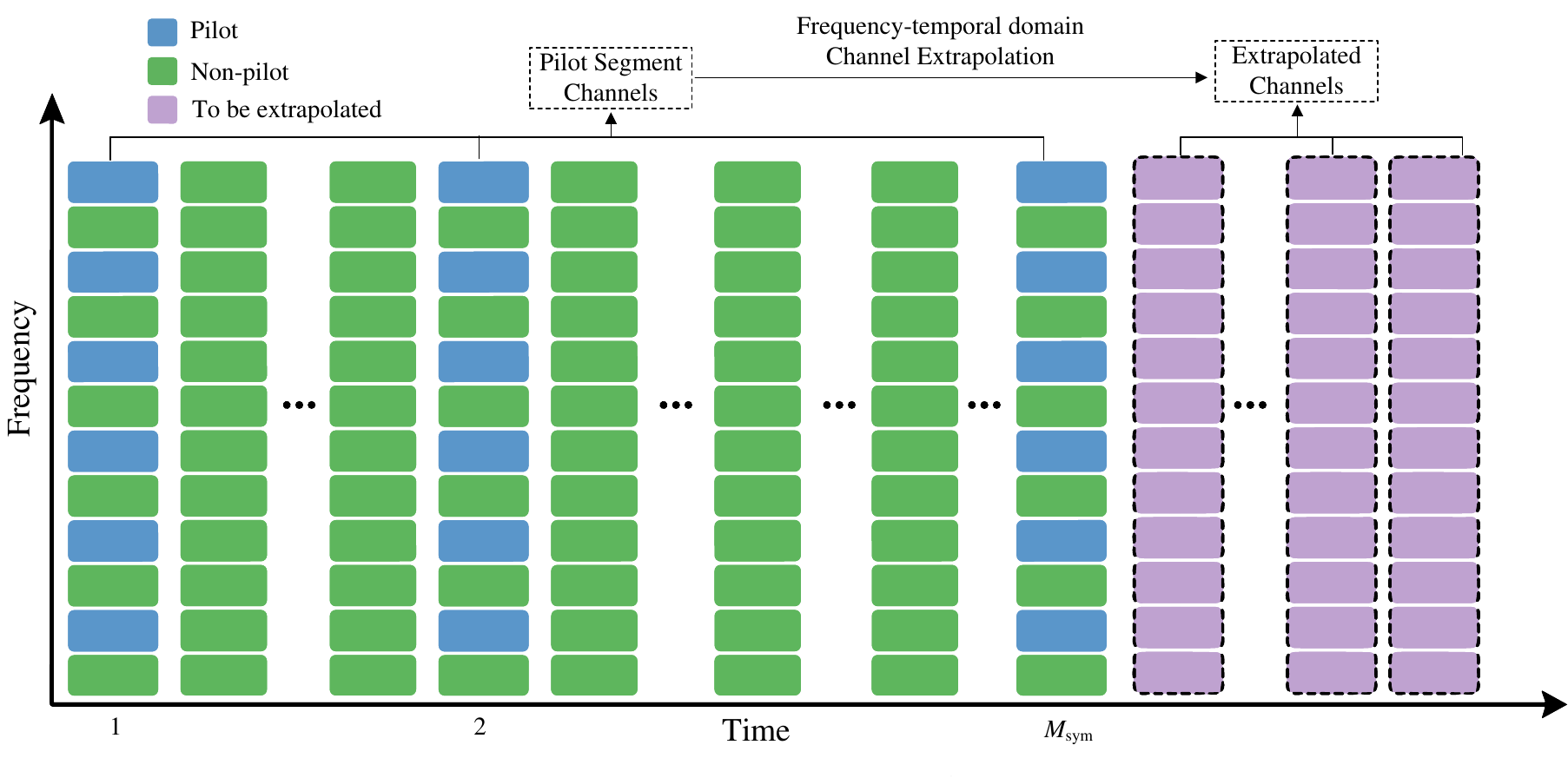}%
	\label{fig:FrequencyTemporalDomainChannelExtrapolationZJW}}
	\caption{Illustration of multi-domain channel extrapolation in upper mid-band massive MIMO-OFDM systems.}
    \vspace{-6mm}
	\label{fig:MultidomainChannelExtrapolationIllustrations}
\end{figure}

Consider a massive MIMO-OFDM system operating in time-division duplexing (TDD) mode, where the BS, equipped with a uniform linear array (ULA) comprising \(N_{\mathrm{an}}\) antennas and \(N_{\mathrm{RF}}\) RF chains with \(N_{\mathrm{RF}} < N_{\mathrm{an}}\), serves a typical single-antenna mobile terminal (MT).
% \footnote{The single-user model is adopted for clarity; multi-user extensions follow by standard orthogonal pilot allocation and per-user processing.}. 
As illustrated in Fig.~\ref{fig:AntennaDomainChannelExtrapolation}, 
the BS adopts an HBF architecture, where $N_{\mathrm{RF}}$ RF chains are connected to $N_{\mathrm{an}}$ antennas through an analog phase shift network and digital baseband processing.
OFDM modulation is employed with $N_{\mathrm{FFT}}$ subcarriers, among which $N_{\mathrm{sc}}$ subcarriers are used for data transmission with a subcarrier spacing of $\Delta f$ and $N_{\mathrm{CP}}$-length cyclic prefix (CP). Hence, the system sampling interval, the OFDM symbol duration, and the CP duration are given by $\Delta T_{\mathrm{sam}}=\frac{1}{N_{\mathrm {FFT}} \Delta f}$, $\Delta T_{\mathrm{sym}} = N_{\mathrm {FFT}} \Delta T_{\mathrm{sam}}$, and $\Delta T_{\mathrm{CP}} = N_{\mathrm{CP}}\Delta T_{\mathrm{sam}}$, respectively, so that the total OFDM symbol duration is $\Delta T=\Delta T_{\mathrm {sym}}+\Delta T_{\mathrm {CP}}$.

%加一段写多域观测抽取这个事，
As antenna arrays and bandwidths continue to scale, exhaustive pilot-based CSI acquisition becomes increasingly costly in next-generation wireless networks. To reduce the resulting pilot overhead,  we adopt sparse channel sounding over the SFT domains and recover or predict the unobserved CSI through multi-domain channel extrapolation. In the spatial domain, the $N_{\mathrm{RF}}$ RF chains are connected to $N_{\mathrm{RF}}$ uniformly selected antenna ports, such that  the number of sounded antennas matches the available RF chains. The uniform antenna selection avoids configuration-dependent antenna-selection signaling and facilitates model generalization. In the frequency domain, the MT employs a comb-type pilot pattern with a uniform spacing of $N_{\mathrm{f}}$ subcarriers, corresponding to a pilot spacing of $\Delta \bar{f}=N_{\mathrm{f}}\Delta f$. In the temporal domain, pilot observations are collected over $N_{\mathrm{t}}$ OFDM symbols, forming a temporal observation window of duration $\Delta \bar{T}=N_{\mathrm{t}}\Delta T$. Under such SFT-domain pilot configuration, as illustrated in Fig.~\ref{fig:MultidomainChannelExtrapolationIllustrations}, multi-domain channel extrapolation is employed to reconstruct the CSI over all $N_{\mathrm{an}}$ antennas and $N_{\mathrm{sc}}$ subcarriers, and to predict the CSI of OFDM symbols beyond the observation window.

\subsection{Channel Model}
The channel impulse response (CIR) at the $n_{\mathrm{an}}$-th antenna is expressed as \cite{hou2024jointCIR}
\begin{equation}\label{eq:CIR_ULA}
	h_{n_{\mathrm{an}}}(t,\tau)
	=\sum_{l=1}^{L} g_{l} e^{j2\pi \nu_{l} t}
	\delta\!\left(\tau-\tau_{n_{\mathrm{an}},l}\right),
\end{equation}
where $L$ is the number of multipath components, and $\tau_{n_{\mathrm{an}},l}=\tau_l+\frac{(n_{\mathrm{an}}-1)d\cos(\vartheta_l)}{c}$ is the propagation delay of the $l$-th path from the MT to the $n_{\mathrm{an}}$-th BS antenna under the far-field planar-wave assumption.
Here, $\vartheta_l$ denotes the angle of arrival (AoA) of the $l$-th path, while $g_l$, $\tau_l$, and $\nu_l$ denote its complex gain, delay, and Doppler frequency, respectively. In addition, $c$ is the speed of light, and $d$ is the inter-antenna spacing.

By taking the Fourier transform of the CIR and stacking the channel responses across antennas, subcarriers, and OFDM symbols, the full SFT-domain channel tensor is obtained as
\begin{equation}\label{eq:SFT_sum_rank1}
	\boldsymbol{\mathcal{H}}=\sum_{l=1}^{L} g_{l}
	\mathbf{a}(\psi_{l})\circ 
	\mathbf{b}(\tau_{l})\circ \mathbf{c}(\nu_{l}),
\end{equation}
where $\psi_l \triangleq d\cos(\vartheta_l)/\lambda$ is the directional cosine of the $l$-th propagation path, and $\lambda$ is the carrier wavelength. The vectors $\mathbf{a}(\psi)$, $\mathbf{b}(\tau)$, and $\mathbf{c}(\nu)$ denote the steering vectors in the angle, delay, and Doppler domains, respectively, with
$[\mathbf{a}(\psi)]_{n_{\mathrm{an}}} \triangleq e^{-j2\pi(n_{\mathrm{an}}-1)\psi}$,
$[\mathbf{b}(\tau)]_{n_{\mathrm{sc}}} \triangleq e^{-j2\pi(n_{\mathrm{sc}}-1)\Delta f\tau}$, and
$[\mathbf{c}(\nu)]_{n_{\mathrm{sym}}} \triangleq e^{j2\pi(n_{\mathrm{sym}}-1)\Delta T\nu}$. The parametric decomposition in \eqref{eq:SFT_sum_rank1} is not specific to FR3 and remains applicable to FR1 and FR2, with the principal band-dependent differences reflected in the statistical distributions of the path gains, delays, angles, and Doppler frequencies. Available measurement results show that FR3 generally exhibits lower propagation loss and richer resolvable multipath than FR2, while tending to exhibit smaller delay and angular spreads and, for a given user velocity, larger Doppler shifts than FR1~\cite{MiaoFR3Measurement,LiuFR3Measurement}.

The representation in \eqref{eq:SFT_sum_rank1} reveals that the
SFT-domain channel admits a rank-$L$ canonical polyadic structure, owing to the limited
scattering characteristics of practical propagation environments. When
this low-rank structure is exploited through an unstructured canonical polyadic
decomposition, the practical lack of uniqueness guarantee calls for
additional structural priors on the factor matrices.  To this end, we construct the factor matrices
$\mathbf{A}(\overline{\boldsymbol{\psi}})$,
$\mathbf{B}(\overline{\boldsymbol{\tau}})$, and
$\mathbf{C}(\overline{\boldsymbol{\nu}})$ of the
SFT-domain channel based on predefined ADD-domain grids, defined as
\begin{equation}
	\begin{aligned}
		\mathbf{A}(\overline{\boldsymbol{\psi}}) & =\left[\mathbf{a}\left(\bar{\psi}_1\right), \ldots, \mathbf{a}\left(\bar{\psi}_{K_{\mathrm{ang}}}\right)\right]\in
        \mathbb{C}^{N_{\mathrm{an}}\times K_{\mathrm{ang}}}, \\
		\mathbf{B}(\overline{\boldsymbol{\tau}}) & =\left[\mathbf{b}\left(\bar{\tau}_1\right), \ldots, \mathbf{b}\left(\bar{\tau}_{K_{\mathrm{de}}}\right)\right]\in
        \mathbb{C}^{N_{\mathrm{sc}}\times K_{\mathrm{de}}}, \\
		\mathbf{C}(\overline{\boldsymbol{\nu}}) & =\left[\mathbf{c}\left(\bar{\nu}_1\right), \ldots, \mathbf{c}\left(\bar{\nu}_{K_{\mathrm{do}}}\right)\right]\in
        \mathbb{C}^{N_{\mathrm{sym}}\times K_{\mathrm{do}}},
	\end{aligned}
\end{equation}
where $K_{x}$, $x \in\{\mathrm{ang}, \mathrm{de}, \mathrm{do}\}$, denotes the number of grid points in the corresponding domain, and $\overline{\boldsymbol{\psi}}$, $ \overline{\boldsymbol{\tau}}$, and $\overline{\boldsymbol{\nu}}$ denote the angle, delay, and Doppler grids, respectively. Following common practice in dynamic grid modeling, uniform sampling is adopted in the angle, delay, and Doppler domains, with grid sizes proportional to the numbers of antennas, subcarriers, and OFDM symbols so as to balance computational complexity and model flexibility.

Accordingly, the SFT-domain channel can be represented by the following Tucker-decomposition-based model~\cite{HouTuckerBasedChannelPrediction}:
\begin{equation}\label{eq:tucker}
\boldsymbol{\mathcal{H}}=\boldsymbol{\mathcal{G}} \times_{1}\mathbf{A}  \times_{2} \mathbf{B} \times_{3} \mathbf{C},
\end{equation}
where $\boldsymbol{\mathcal{G}} \in \mathbb{C}^{K_{\mathrm{ang}} \times K_{\mathrm{de}} \times K_{\mathrm{do}}}$ denotes the ADD-domain channel tensor, whose entries represent the gains associated with the corresponding ADD-domain grid points and each of which contributes a rank-one component in the SFT domains.
    \vspace{-2.5mm}
\subsection{Signal Model}
During channel sounding, sparse measurements are collected over the SFT domains to reduce pilot overhead and CSI acquisition burden. Specifically, the subsampled SFT-domain channel tensor is obtained as
\begin{equation}
\begin{aligned}
	\boldsymbol{\mathcal{H}}_{\mathrm{p}}
	&= \boldsymbol{\mathcal{H}} \times_1 \mathbf{P}_{\mathrm{s}}
	\times_2 \mathbf{P}_{\mathrm{f}}
	\times_3 \mathbf{P}_{\mathrm{t}} \\
	&= \boldsymbol{\mathcal{G}} \times_1 \mathbf{P}_{\mathrm{s}} \mathbf{A}
	\times_2 \mathbf{P}_{\mathrm{f}} \mathbf{B}
	\times_3 \mathbf{P}_{\mathrm{t}} \mathbf{C},
\end{aligned}
\end{equation}
where $\mathbf{P}_{\mathrm{s}} \in \{0,1\}^{M_{\mathrm{an}} \times N_{\mathrm{an}}}$, 
$\mathbf{P}_{\mathrm{f}} \in \{0,1\}^{M_{\mathrm{sc}} \times N_{\mathrm{sc}}}$, and 
$\mathbf{P}_{\mathrm{t}} \in \{0,1\}^{M_{\mathrm{sym}} \times N_{\mathrm{sym}}}$ denote the binary selection matrices in the spatial, frequency, and temporal domains, respectively, $M_{\mathrm{an}}$, $M_{\mathrm{sc}}$, and $M_{\mathrm{sym}}$ denote the numbers of observed antenna ports, subcarriers, and  OFDM symbols, respectively, with $M_{\mathrm{an}}=N_{\mathrm{RF}}=N_{\mathrm{an}}/N_{\mathrm{s}}$, $M_{\mathrm{sc}}=N_{\mathrm{sc}}/N_{\mathrm{f}}$, and $M_{\mathrm{sym}}=N_{\mathrm{sym}}/N_{\mathrm{t}}$. Following the uniform antenna selection, comb-type frequency pilots, and uniform temporal observation pattern described above, the selection matrices are constructed as
\begin{equation}
[\mathbf{P}_{\xi}]_{m,n}=
\begin{cases}
1, & n=1+(m-1)N_\xi,\\
0, & \text{otherwise},
\end{cases}
\quad \xi \in \{\mathrm{s},\mathrm{f},\mathrm{t}\}.
\end{equation}

Unlike the spatial and frequency domains, where extrapolation reconstructs the unobserved antenna and subcarrier entries within the same SFT block, the temporal domain is used for forward prediction. Specifically, the observations within the temporal window $\mathcal{T}_{\mathrm{obs}}$ are exploited to predict the CSI of future OFDM symbols in $\mathcal{T}_{\mathrm{pred}}$, while the CSI of symbols between adjacent pilot-bearing symbols is not reconstructed.

After cyclic-prefix removal and OFDM demodulation, the received pilot observations can be expressed as
\begin{equation}
\boldsymbol{\mathcal{Y}}
=
\boldsymbol{\mathcal{S}}_{\mathrm{p}} \odot \boldsymbol{\mathcal{H}}_{\mathrm{p}}
+
\boldsymbol{\mathcal{Z}},
\end{equation}
where $\boldsymbol{\mathcal{S}}_{\mathrm{p}}$, $\boldsymbol{\mathcal{H}}_{\mathrm{p}}$, and $\boldsymbol{\mathcal{Z}}$ denote the pilot, the observed SFT-domain channel, and the additive white Gaussian noise (AWGN), respectively. Since the pilot symbols are known at both BS and MT, $\boldsymbol{\mathcal{S}}_{\mathrm{p}}$ is assumed to be an all-ones tensor without loss of generality.

By substituting the Tucker-based channel model into the observation equation, the received signal can be expressed in terms of the ADD-domain channel as
\begin{equation}\label{eq:obs_tucker}
\boldsymbol{\mathcal{Y}}
=
\boldsymbol{\mathcal{G}}
\times_1 \mathbf{A}_{\mathrm{o}}
\times_2 \mathbf{B}_{\mathrm{o}}
\times_3 \mathbf{C}_{\mathrm{o}}
+
\boldsymbol{\mathcal{Z}},
\end{equation}
where $\mathbf{A}_{\mathrm{o}} \triangleq \mathbf{P}_{\mathrm{s}}\mathbf{A}$, $\mathbf{B}_{\mathrm{o}} \triangleq \mathbf{P}_{\mathrm{f}}\mathbf{B}$, and $\mathbf{C}_{\mathrm{o}} \triangleq \mathbf{P}_{\mathrm{t}}\mathbf{C}$ denote the subsampled factor matrices associated with the spatial, frequency, and temporal domains, respectively.

Based on \eqref{eq:obs_tucker}, multi-domain channel extrapolation is formulated as the recovery of the ADD-domain channel from the noisy and subsampled SFT-domain pilot observations $\boldsymbol{\mathcal{Y}}$, which contains only partial channel information after the action of the multi-domain selection matrices. Once the ADD-domain channel and its angle, delay, and Doppler structures are estimated, the full SFT-domain channel can be reconstructed via the multilinear transformation in \eqref{eq:tucker}. In this framework, the angle and delay domains support reconstruction over the spatial and frequency domains, respectively, while the Doppler domain enables prediction over the temporal domain, forming a unified framework for multi-domain channel extrapolation.

\section{Problem Formulation}
\label{sec_problem}
In Tucker-based representations, the ADD-domain channel acquisition  depends on its sparse structure. However, 
practical channel sounding provides only sparse SFT-domain observations, whose selection operations alter the effective ADD-domain representation through the subsampled factor matrices. Under uniform multi-domain sampling, this effect can fold multiple physically distinct ADD-domain components into the same observed coefficients, thereby obscuring the  sparse support and introducing structural ambiguity for ADD-domain recovery.  In this section, we formalize the folding effect induced by uniform multi-domain sampling, characterize the resulting asymmetric ambiguity
across the angle, delay, and Doppler domains, and formulate the recovery problem as a support-prior-assisted ADD-domain de-aliasing problem.
    \vspace{-2.5mm}
\subsection{ADD-Domain Aliasing under SFT Decimation}
The effect of SFT-domain pilot decimation on the ADD-domain  can be analyzed through the decimation of  underlying steering vectors. Specifically, decimation along the spatial, frequency, and temporal dimensions is equivalent to subsampling the antenna-, subcarrier-, and OFDM-symbol-indexed steering vectors, respectively.
We therefore first recall the classical one-dimensional decimation property and then extend it to the tensor modes of the Tucker-based ADD representation.
\begin{lemma}\label{lemma:decimation}
Let $x[n]$, $n\in\mathbb{Z}$, be a discrete sequence with discrete-time Fourier transform (DTFT) $X(e^{j\omega})$, where $\omega\in[-\pi,\pi)$. For a positive integer $R$, define the $R$-fold uniform decimation by $y[m]=x[mR]$, $m\in\mathbb{Z}$. The DTFT of $y[m]$ is
\begin{equation}\label{eq:decimation_lemma}
Y(e^{j\omega})=\frac{1}{R}\sum_{r=0}^{R-1}X\!\left(e^{j\frac{\omega-2\pi r}{R}}\right).
\end{equation}
Consequently, decimation compresses the spectrum and generates $R$ equal-weight replicas. Aliasing occurs when the original bandwidth exceeds $\pi/R$.
\end{lemma}
Since the steering responses along the spatial, frequency, and temporal modes take the form of complex exponentials indexed by antennas, subcarriers, and OFDM symbols, respectively, Lemma~\ref{lemma:decimation} indicates that uniform decimation over these indices gives rise to periodic equivalence classes in the corresponding angle, delay, and Doppler parameters. This leads to the following proposition.

\begin{prop}\label{prop:multi_domain_alias}
    Under uniform SFT-domain decimation by $(N_{\mathrm{s}},N_{\mathrm{f}},N_{\mathrm{t}})$, the parameter tuple $(\psi,\tau,\nu)$ and its shifted counterpart generate identical decimated
    observations. Therefore, they satisfy the following equivalence relation:
\begin{equation}\label{eq:alias_equiv}
\scalebox{0.94}{$
(\psi,\tau,\nu)
\sim
(\psi+\frac{p}{N_{\mathrm{s}}},
\tau+\frac{q}{N_{\mathrm{f}}\Delta f},
\nu+\frac{r}{N_{\mathrm{t}}\Delta T}),
\quad p,q,r\in\mathbb Z.
$}
\end{equation}
\end{prop}

\textit{Proof:}
For the spatial domain, for any $p\in\mathbb{Z}$, the $m$-th sampled response satisfies
\begin{equation}
\begin{aligned}
\bigl[\mathbf{P}_{\mathrm{s}}\mathbf{a}(\psi+p/N_{\mathrm{s}})\bigr]_m
&= e^{-j2\pi (m-1)N_{\mathrm{s}}\psi} e^{-j2\pi(m-1)p}  \\
&= \bigl[\mathbf{P}_{\mathrm{s}}\mathbf{a}(\psi)\bigr]_m,
\end{aligned}
\end{equation}
where the last equality follows from $m,p\in\mathbb{Z}$. Hence, the spatial alias equivalence is established. Delay and Doppler cases follow analogously for the decimations by $N_{\mathrm{f}}$ and $N_{\mathrm{t}}$, respectively. This completes the proof.$\hfill \blacksquare$

Proposition~\ref{prop:multi_domain_alias} shows that uniform
SFT-domain decimation renders certain shifted parameter tuples
indistinguishable. To characterize this ambiguity under the discrete signal model, we construct the ADD-domain dictionaries on periodic discrete Fourier transform (DFT) grids, with the grid sizes $K_{\mathrm{ang}}$, $K_{\mathrm{de}}$, and $K_{\mathrm{do}}$ chosen as integer multiples of decimation factors $N_{\mathrm{s}}$, $N_{\mathrm{f}}$, and $N_{\mathrm{t}}$, respectively. The associated reduced grid sizes are defined as $\bar K_{\mathrm{ang}} \triangleq K_{\mathrm{ang}}/N_{\mathrm{s}}$, $\bar K_{\mathrm{de}} \triangleq K_{\mathrm{de}}/N_{\mathrm{f}}$, and $\bar K_{\mathrm{do}} \triangleq K_{\mathrm{do}}/N_{\mathrm{t}}$.
Under this condition, the continuous parameter shifts
in~\eqref{eq:alias_equiv} reduce to cyclic index displacements by
$\bar K_{\mathrm{ang}}$, $\bar K_{\mathrm{de}}$, and
$\bar K_{\mathrm{do}}$ on the corresponding DFT grids.

To describe the discrete counterpart of the equivalence relation in Proposition~\ref{prop:multi_domain_alias}, let $d\in\{\mathrm{ang},\mathrm{de},\mathrm{do}\}$ with $(N_{\mathrm{ang}},N_{\mathrm{de}},N_{\mathrm{do}})=(N_{\mathrm{s}},N_{\mathrm{f}},N_{\mathrm{t}})$. For each $i=1,\ldots,\bar K_d$, define the corresponding one-dimensional aliasing group as
\begin{equation}\label{eq:alias_group_generic}
\mathcal{A}^{d}_{i}
=
\bigl\{
\langle i+p\bar K_d\rangle_{K_d}
\,:\, p=0,\ldots,N_d-1
\bigr\},
\end{equation}
where $\langle n\rangle_K \triangleq ((n{-}1)\bmod K)+1$ denotes the modulo-$K$
indexing operator under 1-based indexing.  Extending this
definition to the joint ADD-domain via the Cartesian product gives
the 3D aliasing group
\begin{equation}\label{eq:3d_alias_group}
\mathcal{A}_{i,j,k}
\triangleq
\{(i',j',k') : i'\in\mathcal{A}^{\mathrm{ang}}_{i},\;
j'\in\mathcal{A}^{\mathrm{de}}_{j},\;
k'\in\mathcal{A}^{\mathrm{do}}_{k}\},
\end{equation}
with $|\mathcal{A}_{i,j,k}|=N_{\mathrm{s}}N_{\mathrm{f}}N_{\mathrm{t}}$. This discrete aliasing structure yields the following proposition on non-identifiability under multi-domain decimation.

\begin{prop}
    \label{prop:3d_folding}
   For the ground-truth ADD-domain channel $\boldsymbol{\mathcal{G}}$ and any
alias-equivalent ADD-domain channel 
$\boldsymbol{\mathcal{G}}'\neq\boldsymbol{\mathcal{G}}$,
    \begin{equation}\label{eq:group_sum_equal}
    \sum_{(i',j',k')\in\mathcal{A}_{i,j,k}}
    [\boldsymbol{\mathcal{G}}]_{i',j',k'}
    =
    \sum_{(i',j',k')\in\mathcal{A}_{i,j,k}}
    [\boldsymbol{\mathcal{G}}']_{i',j',k'},
    \quad \forall\, i,j,k
    \end{equation}
    implies
    \begin{equation}\label{eq:same_obs}
    \boldsymbol{\mathcal{G}}
    \times_1 \mathbf{A}_{\mathrm{o}}
    \times_2 \mathbf{B}_{\mathrm{o}}
    \times_3 \mathbf{C}_{\mathrm{o}}
    \;=\;
    \boldsymbol{\mathcal{G}}'
    \times_1 \mathbf{A}_{\mathrm{o}}
    \times_2 \mathbf{B}_{\mathrm{o}}
    \times_3 \mathbf{C}_{\mathrm{o}}.
    \end{equation}
Hence, $\boldsymbol{\mathcal{G}}$ and its alias-equivalent
 counterpart $\boldsymbol{\mathcal{G}}'$ are observationally
indistinguishable under multi-domain decimation.
    \end{prop}

\textit{Proof:}
The noiseless observation admits the following rank-one expansion over
the ADD-domain grid:
\begin{equation}\label{eq:obs_expand}
\boldsymbol{\mathcal{Y}}_0
=
\sum_{i=1}^{K_{\mathrm{ang}}}
\sum_{j=1}^{K_{\mathrm{de}}}
\sum_{k=1}^{K_{\mathrm{do}}}
[\boldsymbol{\mathcal{G}}]_{i,j,k}\,
\mathbf{a}_{\mathrm{o},i}\circ
\mathbf{b}_{\mathrm{o},j}\circ
\mathbf{c}_{\mathrm{o},k},
\end{equation}
where $\mathbf{a}_{\mathrm{o},i}$, $\mathbf{b}_{\mathrm{o},j}$,
and $\mathbf{c}_{\mathrm{o},k}$ are the $i$-th, $j$-th, and $k$-th
columns of $\mathbf{A}_{\mathrm{o}}$, $\mathbf{B}_{\mathrm{o}}$, and
$\mathbf{C}_{\mathrm{o}}$, respectively. By the aliasing identity in
Proposition~\ref{prop:multi_domain_alias}, the columns of each decimated
dictionary are invariant within the one-dimensional
aliasing groups. Therefore, for any
$i'\in\mathcal{A}^{\mathrm{ang}}_{i}$,
$j'\in\mathcal{A}^{\mathrm{de}}_{j}$, and
$k'\in\mathcal{A}^{\mathrm{do}}_{k}$,
\begin{equation}\label{eq:column_equiv}
\mathbf{a}_{\mathrm{o},i'}=\bar{\mathbf{a}}_{i},\qquad
\mathbf{b}_{\mathrm{o},j'}=\bar{\mathbf{b}}_{j},\qquad
\mathbf{c}_{\mathrm{o},k'}=\bar{\mathbf{c}}_{k},
\end{equation}
where $\bar{\mathbf{a}}_{i}$, $\bar{\mathbf{b}}_{j}$, and
$\bar{\mathbf{c}}_{k}$ denote the common columns associated with the
corresponding aliasing groups.

Since the aliasing sets
$\{\mathcal{A}_{i,j,k}\}_{i=1,j=1,k=1}^{\bar K_{\mathrm{ang}},
\bar K_{\mathrm{de}},\bar K_{\mathrm{do}}}$
form a disjoint partition of the ADD-domain grid, the summation in
\eqref{eq:obs_expand} can be regrouped according to these aliasing sets.
Moreover, by the column equivalence, all entries within the same
$\mathcal{A}_{i,j,k}$ share the identical outer-product factor
$\bar{\mathbf{a}}_i\circ\bar{\mathbf{b}}_j\circ\bar{\mathbf{c}}_k$.
Hence,
\begin{equation}\label{eq:obs_grouped}
\begin{aligned}
\boldsymbol{\mathcal{Y}}_0
&=
\sum_{i=1}^{\bar K_{\mathrm{ang}}}
\sum_{j=1}^{\bar K_{\mathrm{de}}}
\sum_{k=1}^{\bar K_{\mathrm{do}}}
\sum_{(i',j',k')\in\mathcal{A}_{i,j,k}}
[\boldsymbol{\mathcal{G}}]_{i',j',k'}\,
\mathbf{a}_{\mathrm{o},i'}\circ
\mathbf{b}_{\mathrm{o},j'}\circ
\mathbf{c}_{\mathrm{o},k'} \\
&=
\sum_{i=1}^{\bar K_{\mathrm{ang}}}
\sum_{j=1}^{\bar K_{\mathrm{de}}}
\sum_{k=1}^{\bar K_{\mathrm{do}}}
\left(
\sum_{(i',j',k')\in\mathcal{A}_{i,j,k}}
[\boldsymbol{\mathcal{G}}]_{i',j',k'}
\right)
\bar{\mathbf{a}}_i\circ
\bar{\mathbf{b}}_j\circ
\bar{\mathbf{c}}_k,
\end{aligned}
\end{equation}
where the inner sum is the group-wise sum of
$\boldsymbol{\mathcal{G}}$ over the aliasing group
$\mathcal{A}_{i,j,k}$. Therefore, the decimated observation
$\boldsymbol{\mathcal{Y}}_0$ depends on $\boldsymbol{\mathcal{G}}$
only through these group-wise aggregate coefficients. This completes the proof.
 $\hfill\blacksquare$

\begin{remark}\label{rem:aliasing_ambiguity}
    Proposition~\ref{prop:multi_domain_alias} first establishes that uniform
    SFT-domain decimation can make shifted angle-delay-Doppler tuples produce
    identical observations. Proposition~\ref{prop:3d_folding} further reveals
    the ambiguity introduced by uniform multi-domain decimation in the
    ADD domain. Specifically, each ADD-domain aliasing group contains
    $N_{\mathrm{s}}N_{\mathrm{f}}N_{\mathrm{t}}$ aliased positions that are
    observationally indistinguishable from the decimated SFT-domain
    observations. Hence, physically distinct ADD-domain components may give
    rise to the same decimated observations, yielding a many-to-one mapping
    from the original sparse representation to the observations and thereby
    obscuring the original sparse support.
\end{remark}

    \vspace{-4.5mm}
\subsection{Asymmetric Aliasing Characteristics Across Domains}
\label{subsec:asymmetric}
The practical impact of aliasing differs substantially across domains because the delay, Doppler, and angle parameters obey distinct physical constraints.

\subsubsection{Delay Domain}

In the delay domain, channel energy is typically confined to a compact low-delay support, enabling large frequency-domain decimation.
 According to~\cite{3GPP_38901}, about $80\%$ of non-line-of-sight (NLoS) urban macro (UMa) realizations have delay spreads below $1\,\mu\mathrm{s}$. For $\Delta f=60$~kHz, $N_{\mathrm{f}}=8$ yields an unambiguous delay window of $1/(N_{\mathrm{f}}\Delta f)\approx2.08\,\mu\mathrm{s}$, which generally covers the dominant low-delay components. However, when $N_{\mathrm{f}}$ increases to $16$, this window shrinks to about $1.04\,\mu\mathrm{s}$, so high-delay paths may be folded into the low-delay region, as illustrated in the delay-domain row of Fig.~\ref{Fig:SFTDomainDecimationInducedADDFolding}.
Therefore, hard selection of the first 
$\bar{K}_{\mathrm{de}}$ delay bins 
becomes insufficient because it may discard  high-delay 
components and cannot distinguish true low-delay paths from folded 
high-delay replicas.

\subsubsection{Doppler Domain}
The Doppler domain follows a similar low-frequency concentration principle. The maximum Doppler frequency is given by
$\nu_{\max}=f_{\mathrm{c}} v_{\max}/c$, which remains moderate even under high mobility. For example, when $f_{\mathrm{c}}=15$~GHz, $\Delta f=60$~kHz, and $v_{\max}=120$~km/h, we have $\nu_{\max}\approx1.67$~kHz and $\nu_{\max}\Delta T\approx0.03$ per OFDM symbol. Thus, most Doppler energy is confined to the low-frequency region. When the temporal pilot spacing $N_{\mathrm{t}}$ satisfies $\nu_{\max}<1/(2N_{\mathrm{t}}\Delta T)$, high-Doppler components are not folded into the low-Doppler region, which makes low-Doppler bin selection feasible. Accordingly, we retain the low-Doppler bins and set the Doppler grid size to $K_{\mathrm{do}}=S_{\nu}M_{\mathrm{sym}}$, where $S_{\nu}\ge1$ denotes the Doppler oversampling factor.
Nevertheless, low-Doppler concentration does not guarantee an accurate Doppler-domain representation when only a few pilot symbols are available. The limited temporal aperture leads to a coarse Doppler grid, and off-grid Doppler components may leak into neighboring bins, weakening the low-Doppler sparsity.

\subsubsection{Angle Domain}
Unlike the delay and Doppler domains, the angle domain does not exhibit such low-frequency concentration.  The spatial frequency $\psi_l=d\cos(\vartheta_l)/\lambda$ is governed by the AoA, which is determined by the scattering geometry and may occupy arbitrary locations over the Nyquist interval. According to the 3GPP channel model~\cite{3GPP_38901}, multipath AoAs in typical UMa scenarios can span a wide angular range. Therefore, no reliable prior can indicate whether a path should lie in a low- or high-frequency angular grid region.
Under $N_{\mathrm{s}}$-fold spatial decimation, a spatial frequency $\psi$ is indistinguishable from its wrapped replicas
$\psi_r=(\psi+r/N_{\mathrm{s}})\bmod 1$, $r=1,\ldots,N_{\mathrm{s}}-1$,
so the correct alias branch cannot be determined from decimated pilots alone. As shown in Fig.~\ref{Fig:SFTDomainDecimationInducedADDFolding}, unlike delay- or Doppler-domain folding, angle-domain folding introduces aliases over the entire angular grid, degrading spatial extrapolation even for modest decimation factors.

\begin{figure}[!t]
	\centering   
	\includegraphics[width=0.49\textwidth]{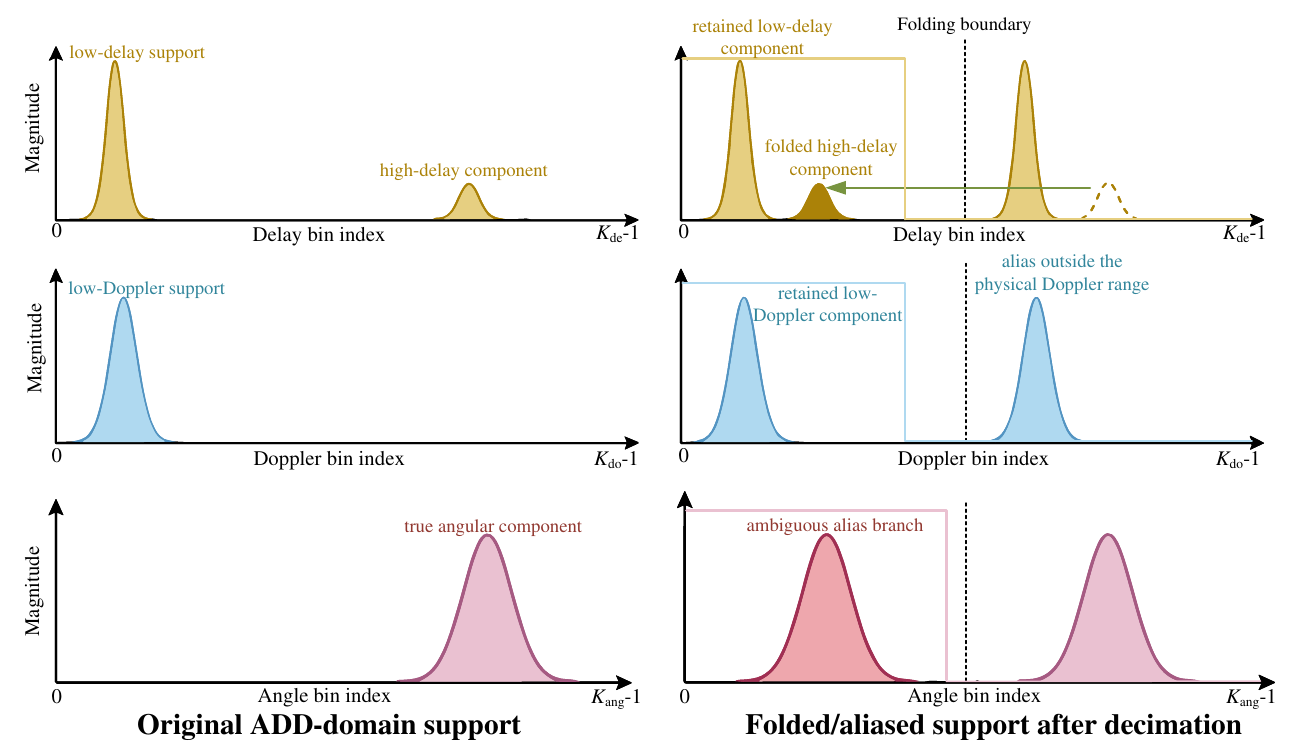} 
	\caption{Illustration of decimation-induced ADD-domain support folding with $N_{\mathrm{s}}=N_{\mathrm{f}}=N_{\mathrm{t}}=2$. The original delay, Doppler, and angle supports are shown together with their folded counterparts after SFT-domain decimation.}
	\label{Fig:SFTDomainDecimationInducedADDFolding} 	
    \vspace{-6mm}
\end{figure}

    \vspace{-3mm}
\subsection{Support-Prior-Assisted De-Aliasing Formulation}
\label{subsec:support_prior}
The preceding analysis shows that the dominant impairment arises from the loss of location-discriminating information within each ADD-domain aliasing group, rather than from intrinsic
similarity among physical components. We therefore exploit ADD-domain priors to compensate for this information loss, resolving delay-domain aliasing and angle-domain ambiguity while mitigating Doppler-domain energy leakage.

In practice, such prior information is often available in modern wireless systems, e.g., from channel knowledge maps (CKMs)~\cite{Citezeng2024CKMtutorialSCSI, CiteZengyongCKMHybridBeamforming}, historical CSI and sensing-assisted perception.
As the delay and angular supports depend on the scattering geometry and user location, they remain stable over many channel coherence intervals and hence require only infrequent acquisition, whereas Doppler priors are governed by the user velocity, which can be estimated from radar echoes in integrated sensing and communication (ISAC)-enabled systems with marginal additional overhead.

To avoid the signaling and computational overhead  associated with acquiring accurate continuous physical parameters, we represent  the prior information as binary support priors. A joint ADD-domain support can capture cross-domain coupling, but its size scales as \(K_{\mathrm{ang}} \times K_{\mathrm{de}} \times K_{\mathrm{do}}\), rendering acquisition, storage, and real-time update costly. Hence, we employ domain-wise marginal support priors, which provide a practical separated representation that retains the essential de-aliasing information in each domain while improving robustness to support errors. Specifically, we introduce
\begin{equation}\label{eq:support_vectors}
    \mathbf{s}_{\mathrm{ang}} \in \{0,1\}^{K_{\mathrm{ang}}\times 1}, 
    \mathbf{s}_{\mathrm{de}}  \in \{0,1\}^{K_{\mathrm{de}}\times 1},  
    \mathbf{s}_{\mathrm{do}}  \in \{0,1\}^{K_{\mathrm{do}}\times 1},
\end{equation}
where \([\mathbf{s}_x]_k = 1\) indicates that the \(k\)-th grid in domain \(x \in \{\mathrm{ang}, \mathrm{de}, \mathrm{do}\}\) is estimated to contain a significant component.

Given the marginal binary support priors, a key issue is how to exploit them for de-aliasing. Hard filtering, which masks unsupported grid points, is non-adaptive and discards useful aliasing information, failing to resolve residual ambiguity under imperfect priors, while model-based regularization requires handcrafted regularizers matched to the three-dimensional aliasing structure and is sensitive to model mismatch. 
Conventional AI-based prior-injection schemes, such as channel concatenation or feature-level fusion, often map external priors into latent features, obscuring their physical semantics and leading to insufficient exploitation of the underlying aliasing structure.
Motivated by this, we reformulate de-aliasing as a prior-conditioned operator learning problem, in which the recovery operator is learned directly from data under the guidance of the support priors. Since the priors serve as auxiliary inputs rather than hard constraints, the formulation tolerates moderate support errors by leveraging reliable support information while suppressing misleading entries.

In this formulation, we use the LS estimate as the network input, defined as
\begin{equation}\label{eq:LS_init}
\boldsymbol{\mathcal{G}}_{\mathrm{LS}}
=
\boldsymbol{\mathcal{Y}}
\times_1 \mathbf{A}_{\mathrm{o}}^{\dagger}
\times_2 \mathbf{B}_{\mathrm{o}}^{\dagger}
\times_3 \mathbf{C}_{\mathrm{o}}^{\dagger}.
\end{equation}
The LS estimate \(\boldsymbol{\mathcal{G}}_{\mathrm{LS}}\)  retains the group-wise aggregate information in the decimated
observation but cannot resolve the corresponding intra-group ambiguity,
thereby providing a coarse initialization for network-based de-aliasing. The recovery problem is therefore cast as estimating \(\boldsymbol{\mathcal{G}}\) from \(\boldsymbol{\mathcal{G}}_{\mathrm{LS}}\) conditioned on the marginal support priors:
\begin{equation}\label{eq:nn_dealiasing}
    \widehat{\boldsymbol{\mathcal{G}}}
    = f_{\boldsymbol{\theta}}\!\left(
    \boldsymbol{\mathcal{G}}_{\mathrm{LS}},\;
    \mathbf{s}_{\mathrm{ang}},\;
    \mathbf{s}_{\mathrm{de}},\;
    \mathbf{s}_{\mathrm{do}}\right),
\end{equation}
where $f_{\boldsymbol{\theta}}$ denotes a neural network that jointly exploits the LS-initialized tensor and the marginal support priors to produce the de-aliased ADD-domain estimate $\widehat{\boldsymbol{\mathcal{G}}}$. The extrapolated SFT-domain channel is subsequently obtained as
\begin{equation}\label{eq:tucker_extrap}
    \widehat{\boldsymbol{\mathcal{H}}}_{\mathrm{extra}}
    = \widehat{\boldsymbol{\mathcal{G}}}
    \times_1 \mathbf{A}(\overline{\boldsymbol{\psi}})
    \times_2 \mathbf{B}(\overline{\boldsymbol{\tau}})
    \times_3 \tilde{\mathbf{C}}(\overline{\boldsymbol{\nu}}),
\end{equation}
where \(\tilde{\mathbf{C}}(\overline{\boldsymbol{\nu}})
=\left[\tilde{\mathbf{c}}\left(\bar{\nu}_1\right), \ldots,
\tilde{\mathbf{c}}\left(\bar{\nu}_{K_{\mathrm{do}}}\right)\right]
\in \mathbb{C}^{N_{\mathrm{pred}} \times K_{\mathrm{do}}}\) denotes
the factor matrix in the temporal domain for channel prediction,
\(N_{\mathrm{pred}}\) is the
prediction length, and \(\tilde{\mathbf{c}}(\bar{\nu})\) denotes the
Doppler-domain steering vector with entry
\([\tilde{\mathbf{c}}(\bar{\nu})]_{n_{\mathrm{pred}}}
= e^{j2\pi\left(T_0 + n_{\mathrm{pred}}\Delta T\right)
\bar{\nu}}\), \(n_{\mathrm{pred}} = 1, \ldots, N_{\mathrm{pred}}\).
Here, \(T_0 = (M_{\mathrm{sym}}-1)N_{\mathrm{t}}\Delta T\) denotes the
prediction origin.

\section{Proposed Network Architecture}
\label{sec_network}

Based on the formulation in~\eqref{eq:nn_dealiasing}, we develop a TANN for ADD-domain de-aliasing. TANN applies axial-attention along the angle, delay, and Doppler domains to capture domain-wise correlations.  To combat multi-domain decimation-induced aliasing, lightweight multi-scale CNN-based gates with  feature-wise linear modulation (FiLM) are employed to inject support priors  into the learned features. Axis-wise gated residual connections are further incorporated to adaptively regulate prior-guided feature refinement.
Moreover, a dual-domain loss is designed to jointly enforce reconstruction consistency.  Owing to the fixed ADD-domain grid, TANN preserves a unified network architecture while maintaining scalability across different SFT-domain decimation configurations. Such configuration scalability provides the basis for mixed-configuration training, enabling robust single-model performance across diverse pilot configurations.

\begin{figure*}[!t]
    \centering
    \includegraphics[width=0.87\textwidth]{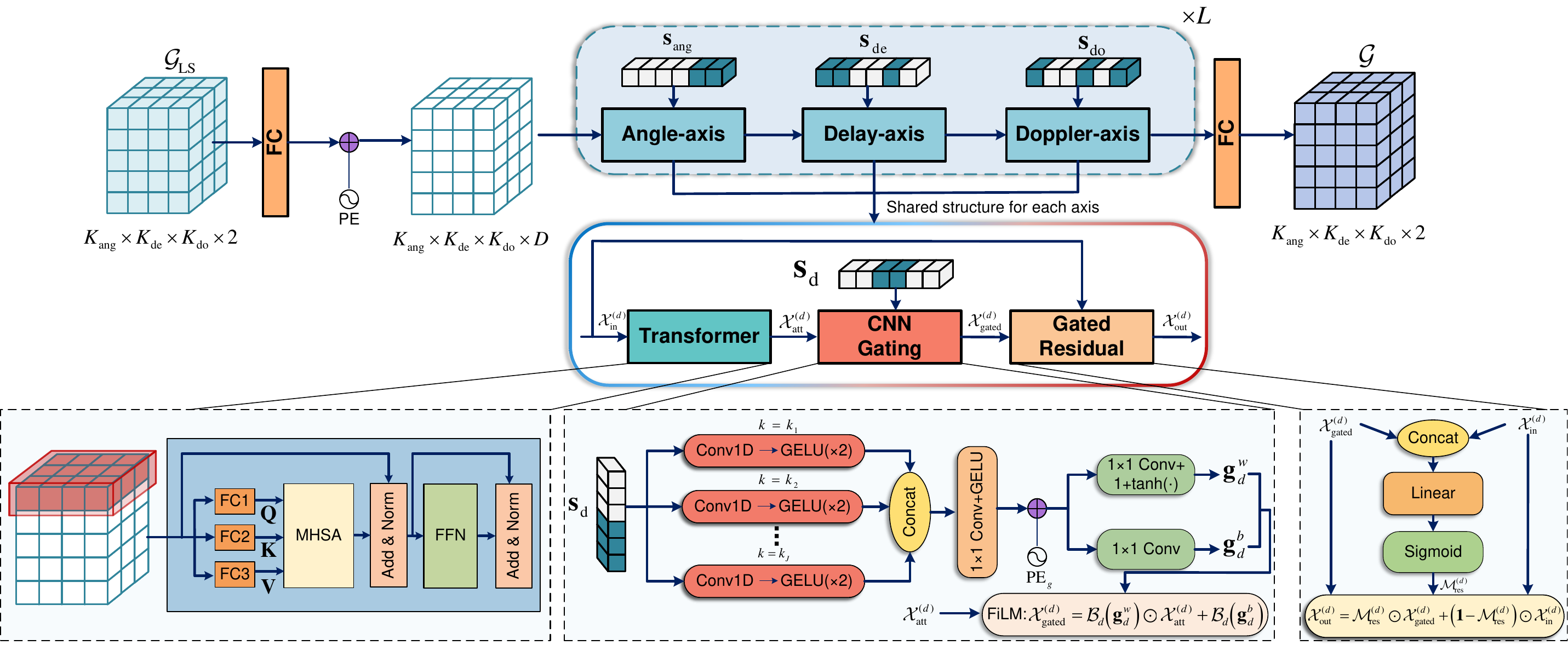}
    \captionsetup{justification=raggedright,singlelinecheck=false}
    \caption{Overall architecture of the proposed TANN.}
    \label{Fig:TANNStructure}
    \vspace{-7mm}
\end{figure*}

    \vspace{-3mm}
\subsection{Overall Architecture}
\label{subsec:embedding}
The overall architecture of TANN is depicted in Fig.~\ref{Fig:TANNStructure}. The network takes the LS-initialized ADD-domain tensor together with domain-wise support priors as inputs, and outputs the de-aliased ADD-domain estimate through three successive stages: a feature embedding with 3D positional encoding, an \(L\)-layer axial-attention backbone, and an output projection followed by multilinear channel reconstruction.

\subsubsection{Feature Embedding and 3D Positional Encoding}
The input
$\boldsymbol{\mathcal{X}}_{\mathrm{in}}\in\mathbb{R}^{K_{\mathrm{ang}}\times
K_{\mathrm{de}}\times K_{\mathrm{do}}\times 2}$ is obtained by stacking
the real and imaginary parts of $\boldsymbol{\mathcal{G}}_{\mathrm{LS}}$.
A point-wise linear layer embeds each grid point into a
$D$-dimensional feature space:
\begin{equation}\label{eq:embedding}
    \boldsymbol{\mathcal{X}}^{(0)}
    =
    \boldsymbol{\mathcal{X}}_{\mathrm{in}}\times_4\mathbf{W}_{\mathrm{emb}}
    +
    \boldsymbol{\mathcal{B}}_{\mathrm{emb}}
    \in
    \mathbb{R}^{K_{\mathrm{ang}}\times K_{\mathrm{de}}\times
    K_{\mathrm{do}}\times D},
\end{equation}
where $\mathbf{W}_{\mathrm{emb}}\in\mathbb{R}^{D\times 2}$ and
\(\boldsymbol{\mathcal{B}}_{\mathrm{emb}}\) denotes the tensor obtained by
broadcasting \(\mathbf{b}_{\mathrm{emb}}\in\mathbb{R}^{D}\) over the first
three modes. 

As discussed in Section~\ref{subsec:support_prior}, resolving ADD-domain aliasing fundamentally requires distinguishing aliased components by their positions in the ADD domains. 
However, Transformer blocks based solely on self-attention are permutation-equivariant and therefore unable to distinguish elements by their positions~\cite{Wang2025DP}. In conventional Transformer architectures, this limitation is typically addressed by one-dimensional (1D) positional encodings. Since the angle, delay, and Doppler axes exhibit distinct positional semantics, we extend the 1D sinusoidal positional encoding to a 3D separable form by superimposing independent axis-wise encodings:
\begin{equation}\label{eq:pe_add}
    \boldsymbol{\mathcal{X}}_{\mathrm{PE}}^{(0)}
    =
    \boldsymbol{\mathcal{X}}^{(0)}
    + \sum_{d\in\{\mathrm{ang},\mathrm{de},\mathrm{do}\}} \mathcal{B}_{d}\!\left(\mathbf{E}_{d}\right),
\end{equation}
where \(\mathbf{E}_{d}\in\mathbb{R}^{K_d\times D}\) is the standard sinusoidal positional encoding along axis \(d\) and \(\mathcal{B}_d(\cdot)\)
denotes the axis-wise
broadcast operator that maps \(\mathbb{R}^{K_d \times C}\) to \(\mathbb{R}^{K_{\mathrm{ang}} \times K_{\mathrm{de}} \times K_{\mathrm{do}} \times C}\) by expansion along the non-\(d\) dimensions.

\subsubsection{Multi-Layer Axial-Attention Backbone}
The encoded features  \(\boldsymbol{\mathcal{X}}_{\mathrm{PE}}^{(0)}\) are then processed by \(L\) axial-attention layers, each applying three axis-wise blocks in sequence:
\begin{equation}\label{eq:backbone}
    \boldsymbol{\mathcal{X}}^{(\ell)}
    =
    f_{\mathrm{do}}^{(\ell)}\!\left(
    f_{\mathrm{de}}^{(\ell)}\!\left(
    f_{\mathrm{ang}}^{(\ell)}\!\left(
    \boldsymbol{\mathcal{X}}^{(\ell-1)};\,\mathbf{s}_{\mathrm{ang}}
    \right);\,\mathbf{s}_{\mathrm{de}}
    \right);\,\mathbf{s}_{\mathrm{do}}
    \right), 
\end{equation}
where $\ell=1,\ldots,L$, \(f_{d}^{(\ell)}(\cdot;\mathbf{s}_d)\) integrates axis-wise multi-head self-attention (MHSA), support-prior-conditioned gating, and a gated residual connection, as detailed in Section~\ref{subsec:axial_layer}.

\subsubsection{Output Projection and SFT-Domain Channel Reconstruction}
After \(L\) axial-attention layers, a point-wise linear projection maps the features back to two channels:
\begin{equation}\label{eq:output}
    \boldsymbol{\mathcal{X}}_{\mathrm{out}}
    =
    \boldsymbol{\mathcal{X}}^{(L)}\times_4\mathbf{W}_{\mathrm{proj}}
    +
    \boldsymbol{\mathcal{B}}_{\mathrm{proj}}
    \in
    \mathbb{R}^{K_{\mathrm{ang}}\times K_{\mathrm{de}}\times
    K_{\mathrm{do}}\times 2},
\end{equation}
where \(\mathbf{W}_{\mathrm{proj}}\in\mathbb{R}^{2\times D}\), and
\(\boldsymbol{\mathcal{B}}_{\mathrm{proj}}\) denotes the tensor obtained by
broadcasting \(\mathbf{b}_{\mathrm{proj}}\in\mathbb{R}^{2}\) over the first
three modes. The resulting tensor is regarded as the de-aliased ADD-domain estimate $\widehat{\boldsymbol{\mathcal{G}}}$, from which the extrapolated SFT-domain channel is reconstructed via the multilinear transformation in~\eqref{eq:tucker_extrap}.

\vspace{-3mm}
\subsection{Axial-Attention Layer}
\label{subsec:axial_layer}

Applying standard MHSA to the flattened 3D tensor would require processing a sequence of length \(K_{\mathrm{ang}} K_{\mathrm{de}} K_{\mathrm{do}}\) with complexity \(\mathcal{O}\!\left((K_{\mathrm{ang}} K_{\mathrm{de}} K_{\mathrm{do}})^2 D\right)\), which is computationally prohibitive in practical systems. To reduce this complexity, we leverage the separable structure across the angle, delay, and Doppler domains and decompose the 3D attention in each of the \(L\) axial-attention layers into three successive 1D axis-wise operations, reducing the complexity to \(\mathcal{O}\!\left(K_{\mathrm{ang}} K_{\mathrm{de}} K_{\mathrm{do}} (K_{\mathrm{ang}} {+} K_{\mathrm{de}} {+} K_{\mathrm{do}}) D\right)\). 
This decomposition is implemented by the shared axis-wise architecture shown in Fig.~\ref{Fig:TANNStructure}. 
Specifically, within each axial-attention layer, the angle-, delay-, and Doppler-wise operations share the same axis-wise block, which consists of a Transformer, a support-conditioned CNN gating module, and a gated residual unit, while being conditioned on \(\mathbf{s}_{\mathrm{ang}}\), \(\mathbf{s}_{\mathrm{de}}\), and \(\mathbf{s}_{\mathrm{do}}\), respectively.

\subsubsection{Axis-Wise MHSA}
For axis \(d\) with length \(K_d\), all \(K_d\)-length sequences along that axis are extracted from the feature tensor \(\boldsymbol{\mathcal{X}}\in \mathbb{R}^{K_{\mathrm{ang}} \times K_{\mathrm{de}} \times K_{\mathrm{do}} \times D}\). In the first axial layer, \(\boldsymbol{\mathcal{X}}\) is initialized as the position-encoded input \(\boldsymbol{\mathcal{X}}_{\mathrm{PE}}^{(0)}\). This extraction yields \(N_d^{\perp} \triangleq \prod_{d' \neq d} K_{d'}\) independent sequences, and MHSA is then applied to them in parallel with \(N_h\) heads of dimension \(D_h = D / N_h\). For each sequence \(\tilde{\mathbf{X}}_d \in \mathbb{R}^{K_d \times D}\), head-specific projection matrices map the features into \(N_h\) representation subspaces. For head \(h\), the query, key, and value representations are
\begin{equation}\label{eq:qkv}
\bigl[\mathbf{Q}_h^{(d)},\mathbf{K}_h^{(d)},\mathbf{V}_h^{(d)}\bigr]
=
\tilde{\mathbf{X}}_d
\bigl[\mathbf{W}_{h,d}^{Q},\mathbf{W}_{h,d}^{K},\mathbf{W}_{h,d}^{V}\bigr],
\end{equation}
where \(\mathbf{W}_{h,d}^{Q}, \mathbf{W}_{h,d}^{K}, \mathbf{W}_{h,d}^{V} \in \mathbb{R}^{D \times D_h}\). Multi-head aggregation is performed via concatenation and linear projection with a learnable output matrix \(\mathbf{W}_d^O \in \mathbb{R}^{D \times D}\):
\begin{align}
    &\mathbf{Y}_h^{(d)} = \mathrm{softmax}\!\left(\mathbf{Q}_h^{(d)} \left(\mathbf{K}_h^{(d)}\right)^{T}/\sqrt{D_h}\right) \mathbf{V}_h^{(d)}
    , \label{eq:head}\\[3pt]
    &\mathrm{Att}^{(d)}(\tilde{\mathbf{X}}_d) = \mathrm{Concat}\!\left(\mathbf{Y}_1^{(d)},\, \ldots,\, \mathbf{Y}_{N_h}^{(d)}\right)\mathbf{W}_d^O, \label{eq:mhsa}
\end{align}
where  \(\mathrm{Concat}(\cdot)\) denotes concatenation along the feature dimension. Following standard Transformer architecture~\cite{vaswani2017attention}, each axis-wise MHSA block is followed by a position-wise feed-forward network (FFN), with pre-layer normalization and residual connections. The processed axis-wise sequences are then folded back into the 4D tensor form to obtain the axis-wise output feature tensor \(\boldsymbol{\mathcal{X}}_{\mathrm{att}}^{(d)} \in \mathbb{R}^{K_{\mathrm{ang}} \times K_{\mathrm{de}} \times K_{\mathrm{do}} \times D}\).

\subsubsection{Multi-Scale CNN Context Gating}
Following each axis-wise MHSA step, the corresponding domain-wise support prior 
\(\mathbf{s}_d\in \{0,1\}^{K_d\times1}\), with 
\(d \in \{\mathrm{ang}, \mathrm{de}, \mathrm{do}\}\), is injected through a lightweight context gating module.
The gate generates soft, data-adaptive multiplicative and additive modulation signals from the support vector. Compared with complex cross-attention or dense feature concatenation schemes, this design introduces minimal overhead while progressively refining features within the serial axial composition.

To generate the gating parameters from \(\mathbf{s}_d\), we employ a multi-scale 1D CNN with \(J\) parallel branches of distinct kernel sizes \(\kappa^{(j)}
\), so as to capture local support features and neighborhood contiguity. Small kernels capture narrow isolated peaks associated with dominant paths, whereas large kernels capture wider contiguous clusters induced by angular or delay spreading. Each branch \(j \in \{1, \ldots, J\}\) contains two stacked 1D convolutional layers with Gaussian error linear unit (GELU) activations, where the two-layer design enlarges the effective receptive field and the successive GELU nonlinearities enhance representation capacity:
\begin{equation}\label{eq:branch}
    \mathbf{h}_d^{(j)} =
    \mathrm{GELU}\!\left(
    \mathrm{Conv}^{\kappa^{(j)}
    }\!\left(
    \mathrm{GELU}\!\left(
    \mathrm{Conv}^{\kappa^{(j)}
    }\!\left(\mathbf{s}_d\right)
    \right)
    \right)\right),
\end{equation}
where $\mathbf{h}_d^{(j)} \in \mathbb{R}^{K_d \times D_g}$, $\mathbf{s}_d$ is treated as a single-feature 1D signal, and $D_g$ denotes the number of intermediate feature maps per branch. The multi-scale branch outputs are then concatenated along the feature dimension and compressed by a point-wise $1 \!\times\! 1$ convolution:
\begin{equation}\label{eq:fusion}
    \mathbf{f}_d =
    \mathrm{GELU}\!\left(
    \mathrm{Conv}^{1}\!\left(
    \left[\mathbf{h}_d^{(1)};\, \ldots ;\, \mathbf{h}_d^{(J)}\right]
    \right)\right)
    \in \mathbb{R}^{K_d \times D_g},
\end{equation}
where the fusion layer input has \(J D_g\) feature maps.

To encode positional semantics within the gate module, a fixed sinusoidal positional encoding \(\mathbf{E}_d^{(g)} \in \mathbb{R}^{K_d \times D_g}\) is added to \(\mathbf{f}_d\) before projection, so that the generated gates are position-aware. We then adopt FiLM-style conditioning, where support-driven features produce a multiplicative scale and an additive bias through two separate point-wise \(1 \!\times\! 1\) convolutional projections:
\begin{align}
    \mathbf{g}_d^{w} &= \sigma_g\!\left(\mathrm{Conv}^{1}\!\left(\mathbf{f}_d + \mathbf{E}_d^{(g)}\right)\right) \in \mathbb{R}^{K_d \times D}, \label{eq:gw}\\
    \mathbf{g}_d^{b} &= \mathrm{Conv}^{1}\!\left(\mathbf{f}_d + \mathbf{E}_d^{(g)}\right) \in \mathbb{R}^{K_d \times D}, \label{eq:gb}
\end{align}
where \(\sigma_g(\mathbf{z}) = 1 + \tanh(\mathbf{z})\) yields a gate centered at one with range
\((0, 2)\).  The modulated output is then written as
\begin{equation}\label{eq:gate_apply}
    \boldsymbol{\mathcal{X}}_{\mathrm{gated}}^{(d)} =
    \mathcal{B}_d\!\left(\mathbf{g}_d^{w}\right) \odot
    \boldsymbol{\mathcal{X}}_{\mathrm{att}}^{(d)} +
    \mathcal{B}_d\!\left(\mathbf{g}_d^{b}\right)
    \in \mathbb{R}^{K_{\mathrm{ang}} \times K_{\mathrm{de}} \times K_{\mathrm{do}} \times D}.
\end{equation}
To further reduce trainable parameters, we adopt layer sharing for the context gate modules, where the prior-guided gates are shared across all \(L\) layers.

\subsubsection{Per-Axis Gated Residual Connection}

Rather than a fixed additive skip connection, we adopt a learnable gated
residual connection that adaptively balances the modulated and original
features.  Specifically, the modulated output \(\boldsymbol{\mathcal{X}}_{\mathrm{gated}}^{(d)}\) is fused with the axis input \(\boldsymbol{\mathcal{X}}_{\mathrm{in}}^{(d)}\) as follows:
\begin{align}
    &\boldsymbol{\mathcal{M}}_{\mathrm{res}}^{(d)} = \sigma\!(
        \mathrm{Concat}\!\left(\boldsymbol{\mathcal{X}}_{\mathrm{in}}^{(d)}, \boldsymbol{\mathcal{X}}_{\mathrm{gated}}^{(d)}\right)
        \times_4\mathbf{W}_{\mathrm{res}}^{(d)} +\boldsymbol{\mathcal{B}}_{\mathrm{res}}^{(d)} ), \label{eq:res_gate} \\ 
        &\boldsymbol{\mathcal{X}}_{\mathrm{out}}^{(d)} = \boldsymbol{\mathcal{M}}_{\mathrm{res}}^{(d)} \odot \boldsymbol{\mathcal{X}}_{\mathrm{gated}}^{(d)} + \left( \mathbf{1} - \boldsymbol{\mathcal{M}}_{\mathrm{res}}^{(d)} \right) \odot \boldsymbol{\mathcal{X}}_{\mathrm{in}}^{(d)} ,
        \label{eq:res_out}
\end{align}
where \(\mathbf{W}_{\mathrm{res}}^{(d)} \in \mathbb{R}^{D \times 2D}\), 
\(\boldsymbol{\mathcal{B}}_{\mathrm{res}}^{(d)}\) is formed by expanding
\(\mathbf{b}_{\mathrm{res}}^{(d)} \in \mathbb{R}^{D}\) along the first three
modes, \(\sigma(\cdot)\) denotes the sigmoid function,
and \(\mathbf{1}\) denotes the all-ones tensor of compatible size. Since the gate is derived from both inputs
jointly, it can  retain the refined features where the context gate
provides reliable modulation and fall back to the original features otherwise.

The output \(\boldsymbol{\mathcal{X}}_{\mathrm{out}}^{(d)}\) of one axis becomes  input
to the next; completing all three axis-wise steps constitutes one axial-attention layer, and the output of the \(L\)-th layer is converted to
\(\widehat{\boldsymbol{\mathcal{G}}}\) via~\eqref{eq:output}.  The
extrapolated SFT-domain channel is then reconstructed via the
multilinear transformation in~\eqref{eq:tucker_extrap}.

\vspace{-2.5mm}
\subsection{SFT-ADD Dual-Domain Consistency Loss Design}
\label{subsec:loss}

Different from SFT-domain supervision adopted in prior extrapolation methods~\cite{3DDomainExtrapolationGaoFeiFei}, our network predicts the ADD-domain tensor $\widehat{\boldsymbol{\mathcal{G}}}$, from which the SFT-domain channel is obtained through~\eqref{eq:tucker_extrap}. Since SFT-domain NMSE alone does not explicitly enforce ADD-domain sparsity, it may lead to suboptimal extrapolation. We therefore adopt a composite loss that combines SFT-domain reconstruction with ADD-domain spectral consistency.

The primary loss measures the NMSE between the extrapolated and ground-truth SFT-domain channels over the prediction set \(\mathcal{T}_{\mathrm{pred}}\):
\begin{equation}\label{eq:loss_nmse}
    \mathcal{L}_{\mathrm{NMSE}}(\boldsymbol{\theta}) = \mathbb{E}\!\left[
        \big\|\widehat{\boldsymbol{\mathcal{H}}}_{\mathrm{extra}}
        - \boldsymbol{\mathcal{H}}_{\mathrm{pred}}\big\|_{\mathrm{F}}^2 \big/
        \big\|\boldsymbol{\mathcal{H}}_{\mathrm{pred}}\big\|_{\mathrm{F}}^2
    \right],
\end{equation}
where \(\boldsymbol{\mathcal{H}}_{\mathrm{pred}} \in \mathbb{C}^{N_{\mathrm{an}} \times N_{\mathrm{sc}} \times N_{\mathrm{pred}}}\) denotes the ground-truth SFT-domain channel tensor corresponding to the OFDM symbols in the prediction set \(\mathcal{T}_{\mathrm{pred}}\).

To guide the network to capture the ADD-domain sparsity structure, we introduce axis-wise spectral consistency terms. Since the proposed network adopts an axial-attention architecture that processes the ADD-domain tensor along each axis independently, formulating the auxiliary loss in an axis-wise manner aligns with the network structure and facilitates the learning of the sparsity pattern along each dimension. Specifically, for each axis \(d \in \{\mathrm{ang}, \mathrm{de}, \mathrm{do}\}\), the marginal power spectrum of  \(\boldsymbol{\mathcal{X}} \in \mathbb{C}^{K_{\mathrm{ang}} \times K_{\mathrm{de}} \times K_{\mathrm{do}}}\) is  defined as
$
    p_{k_d}^{(d)}\big(\boldsymbol{\mathcal{X}}\big)
    = \sum_{\mathbf{k}_{\bar{d}}}
    |[\boldsymbol{\mathcal{X}}]_{k_{\mathrm{ang}},k_{\mathrm{de}},k_{\mathrm{do}}}|^2
$, where \(\mathbf{k}_{\bar{d}}\) collects indices on the remaining two axes. The corresponding spectral vector is \(\mathbf{p}^{(d)}(\boldsymbol{\mathcal{X}}) = [p_1^{(d)}(\boldsymbol{\mathcal{X}}), \ldots, p_{K_d}^{(d)}(\boldsymbol{\mathcal{X}})]^{T} \in \mathbb{R}_{+}^{K_d}\). The spectral consistency along axis \(d\) is quantified by the cosine similarity
\begin{equation}\label{eq:spec_similarity}
    \rho^{(d)}(\widehat{\boldsymbol{\mathcal{G}}}, \boldsymbol{\mathcal{G}})
    = \frac{
        \langle \mathbf{p}^{(d)}(\widehat{\boldsymbol{\mathcal{G}}}),\,
        \mathbf{p}^{(d)}(\boldsymbol{\mathcal{G}}) \rangle
    }{
        \| \mathbf{p}^{(d)}(\widehat{\boldsymbol{\mathcal{G}}}) \|_2\,
        \| \mathbf{p}^{(d)}(\boldsymbol{\mathcal{G}}) \|_2
    },
\end{equation}
and the auxiliary spectral loss on axis \(d\) is \(\mathcal{L}_{\mathrm{spec}}^{(d)}(\boldsymbol{\theta}) = 1 - \rho^{(d)}(\widehat{\boldsymbol{\mathcal{G}}}, \boldsymbol{\mathcal{G}})\). The overall training objective is
\begin{equation}\label{eq:loss_total}
    \mathcal{L}(\boldsymbol{\theta}) = w_{\mathrm{main}}\,
    \mathcal{L}_{\mathrm{NMSE}} + \sum_{d \in \{\mathrm{ang},\, \mathrm{de},\, \mathrm{do}\}} w_d\,
    \mathcal{L}_{\mathrm{spec}}^{(d)},
\end{equation}
where auxiliary weights \(w_d\) are gradually reduced via exponential decay during training. This composite loss guides the network toward sparse ADD-domain representations in the early phase, while progressively shifting focus to the primary NMSE objective to refine end-to-end extrapolation accuracy.

\vspace{-2mm}
\subsection{Mixed-Configuration Training and  Complexity Analysis}

\subsubsection{Mixed-Configuration Training}
The ADD-domain representation decouples the network input dimensionality from the SFT-domain decimation factors, since its tensor dimensions are specified only by the predefined angle, delay, and Doppler grids. This decoupling allows the proposed network to accommodate different decimation configurations under a unified architecture, and consequently enables mixed-configuration training. Specifically, the decimation factors and signal-to-noise ratios (SNRs) are randomized across samples within each mini-batch. Such configuration-invariant dimensionality is not available in prior SFT-domain designs~\cite{3DDomainExtrapolationGaoFeiFei}, where the tensor dimensions vary with the decimation factors and thus require configuration-specific training. As a result, the proposed framework reduces retraining and deployment overhead while achieving robust single-model performance across different decimation configurations and noise levels.
Moreover, the network learns a prior-conditioned de-aliasing mapping from channel samples generated under 3GPP-defined outdoor propagation scenarios and their associated support priors. During deployment, the support-prior input can be derived from the specific propagation conditions observed at the target site, enabling site-aware adaptation without site-wise retraining.

\subsubsection{Complexity Analysis}

The computational complexity of the proposed framework is dominated by the multilinear transformations and the TANN inference. The multilinear transformations, governed by mode-wise tensor-matrix multiplications in \eqref{eq:LS_init} and \eqref{eq:tucker_extrap}, incur \(\mathcal{O}(K_{\mathrm{ang}}K_{\mathrm{de}}K_{\mathrm{do}}\bar{M})\) and \(\mathcal{O}(K_{\mathrm{ang}}K_{\mathrm{de}}K_{\mathrm{do}}\bar{N})\) operations, where \(\bar{M}\triangleq M_{\mathrm{an}}+M_{\mathrm{sc}}+M_{\mathrm{sym}}\) and \(\bar{N}\triangleq N_{\mathrm{an}}+N_{\mathrm{sc}}+N_{\mathrm{pred}}\), whereas the TANN is governed by its \(L\)-layer axial-attention backbone with cost \(\mathcal{O}(LK_{\mathrm{ang}}K_{\mathrm{de}}K_{\mathrm{do}}(K_{\mathrm{ang}}+K_{\mathrm{de}}+K_{\mathrm{do}})D)\), while the multi-scale CNN context-gating, gated-residual, embedding, and projection layers are point-wise and thus negligible. Retaining the dominant terms, the overall complexity of the proposed framework is
$
\mathcal{O}\!\left(K_{\mathrm{ang}}K_{\mathrm{de}}K_{\mathrm{do}}\left(\bar{M}+\bar{N}+L(K_{\mathrm{ang}}+K_{\mathrm{de}}+K_{\mathrm{do}})D\right)\right).
$

\section{Simulation Results}
\label{sec_simulation}

\subsection{Simulation Configuration}

\subsubsection{Scenario Setting}

In this section, we present simulation results to evaluate the performance of the proposed scheme. To generate the channels for simulations, we employ the QuaDRiGa channel simulator~\cite{Cite2014Quadriga}, which is capable of producing time-varying massive MIMO-OFDM channels consistent with the 3GPP 38.901 specifications~\cite{3GPP_38901} and has been validated in various field trials.  A total of \(30{,}000\) channel samples are generated under mixed SNRs and decimation factors, and then partitioned into \(24{,}000\), \(3{,}000\), and \(3{,}000\) samples for training, validation, and testing, respectively. Unless otherwise specified, the simulation configurations follow the system model in Section~\ref{sec_system}. The basic simulation parameters and TANN hyper-parameters are jointly summarized in Table~\ref{tab:simulation_settings}.

\begin{table}[!t]
    \small
    \centering
    \caption{Summary of Simulation Settings}
    \label{tab:simulation_settings}
    \scalebox{0.9}{
    \begin{tabular}{c|c}
        \hline
        Parameter & Values \\ 
        \hline

        \multicolumn{2}{c}{Basic System Parameters} \\
        \hline
        Carrier frequency $f_{\mathrm{c}}$ (GHz) & 10, 15 (training), 20 \\ 
        Subcarrier spacing $\Delta f$ & 60 kHz  \\ 
        Number of subcarriers $N_{\mathrm{sc}}$ & 64 \\ 
        Number of pilot interval symbols $N_{\mathrm{t}}$ & 14 \\ 
        Number of pilot symbols $M_{\mathrm{sym}}$ & 10 \\ 
        Number of prediction symbols $N_{\mathrm{pred}}$ & 14 \\ 
        Number of antennas $N_{\mathrm{an}}$ & 32 \\ 
        Frequency decimation factor $N_{\mathrm{f}}$ & 2, 4, 8, 16 \\ 
        Spatial decimation factor $N_{\mathrm{s}}$ & 1, 2, 4 \\ 
        MT velocity $v$ (km/h) & 3, 30, 60 (training), 90 \\ 
        Doppler oversampling factor $S_{\nu}$ & 2 \\
        Channel model & 38.901 RMa NLoS \\
        \hline

        \multicolumn{2}{c}{TANN Hyper-Parameters} \\
        \hline
        Initial learning rate & $1 \times 10^{-3}$ \\ 
        Batch size & 16 \\ 
        Multi-scale convolution kernel sizes $\kappa^{(j)}$ & 1, 3 \\ 
        Embedding dimension $D$ & 16 \\
        Number of attention heads $N_h$ & 4 \\
        Number of axial-attention layers $L$ & 4 \\
        \hline
    \end{tabular}
    }
    \vspace{-6mm}
\end{table}

\subsubsection{Benchmarks}
To demonstrate the superiority of the proposed support-prior-assisted TANN
(SPA-TANN), we compare it with several representative channel extrapolation benchmarks. Guided by the aliasing analysis in Section~\ref{sec_problem},
SPA-TANN jointly integrates the multi-domain support priors into the network
inference process for learned de-aliasing, whereas the prior-assisted (PA)-prefixed benchmarks
employ angle-domain circular shifting to resolve the aliasing ambiguity. The considered baselines are described as follows:
\begin{itemize}
    \item \textbf{PA-OMP~\cite{CiteOMPBaseLine}}: The orthogonal matching pursuit (OMP) algorithm estimates the ADD-domain channel from \(\boldsymbol{\mathcal{Y}}\) via greedy sparse recovery over a predefined dictionary.

\item \textbf{PA-VSD~\cite{CiteVSDTensorDecompositionMethod}}: The Vandermonde structured decomposition (VSD) algorithm estimates the channel parameters from \(\boldsymbol{\mathcal{Y}}\) via tensor decomposition by exploiting the Vandermonde structure of the factor matrices in \eqref{eq:tucker}.
	
\item \textbf{KDD-SFTCEN~\cite{3DDomainExtrapolationGaoFeiFei}}:
The  knowledge-and-data-driven SFT channel extrapolation network (KDD-SFTCEN) performs channel extrapolation in the SFT domains without using prior information.

	\item \textbf{PA-KDD-SFTCEN~\cite{3DDomainExtrapolationGaoFeiFei}}:
    PA-KDD-SFTCEN is the variant of KDD-SFTCEN that performs channel extrapolation in the SFT domains after PA-based preprocessing.
\end{itemize}

\vspace{-3mm}
\subsection{Performance Evaluation}
\subsubsection{NMSE Versus SNR} 
Fig.~\ref{fig:sim_NMSE_SNR} compares the NMSE performance of all schemes versus SNR under two pilot-decimation configurations.
In Fig.~\ref{fig:sim_NMSE_SNR_1s2f}, with $(N_{\mathrm{s}}, N_{\mathrm{f}})=(1,2)$, there is no spatial decimation and  therefore no angle-domain aliasing arises. In this regime, both KDD-SFTCEN and PA-KDD-SFTCEN achieve comparable performance and significantly outperform the model-based PA-OMP and PA-VSD benchmarks, confirming that data-driven methods can effectively exploit the SFT correlations in the absence of angular ambiguity.  
SPA-TANN further improves upon PA-KDD-SFTCEN, indicating that its ADD-domain formulation provides more effective multi-domain feature modeling.
The performance hierarchy changes markedly in Fig.~\ref{fig:sim_NMSE_SNR_2s4f} under $(N_{\mathrm{s}}, N_{\mathrm{f}})=(2,4)$, where two-fold spatial decimation introduces severe angle-domain aliasing. Without PA-based preprocessing, KDD-SFTCEN fails to resolve the
aliasing ambiguity and degrades below PA-VSD, which benefits from
PA-based de-aliasing. This reveals that angular aliasing induces a
fundamental ambiguity that purely data-driven learning cannot overcome.
Incorporating the PA-based preprocessing (PA-KDD-SFTCEN) restores
advantage over model-based methods, while the proposed SPA-TANN further extends
this margin through its context-gated prior injection for ADD-domain de-aliasing.

\begin{figure*}[!t]
    \centering
    \begin{minipage}[t]{0.492\textwidth}
        % \vspace{0pt}
        \centering

        \subfloat[$N_{\mathrm{s}}=1$, $N_{\mathrm{f}}=2$]{
            \includegraphics[width=0.47\linewidth]{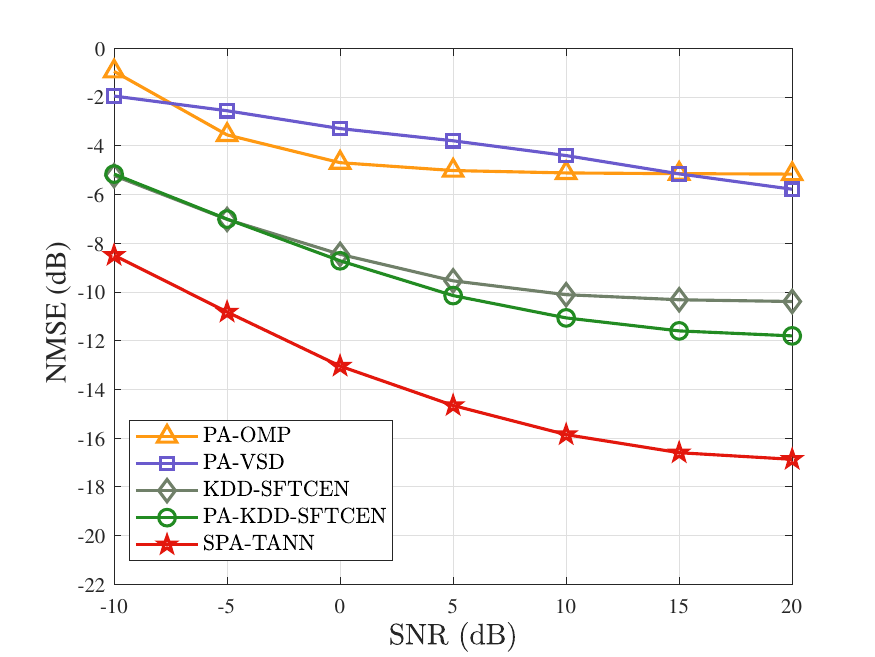}
            \label{fig:sim_NMSE_SNR_1s2f}
        }%
        % \hspace{0.001\linewidth}
        \subfloat[$N_{\mathrm{s}}=2$, $N_{\mathrm{f}}=4$]{
            \includegraphics[width=0.47\linewidth]{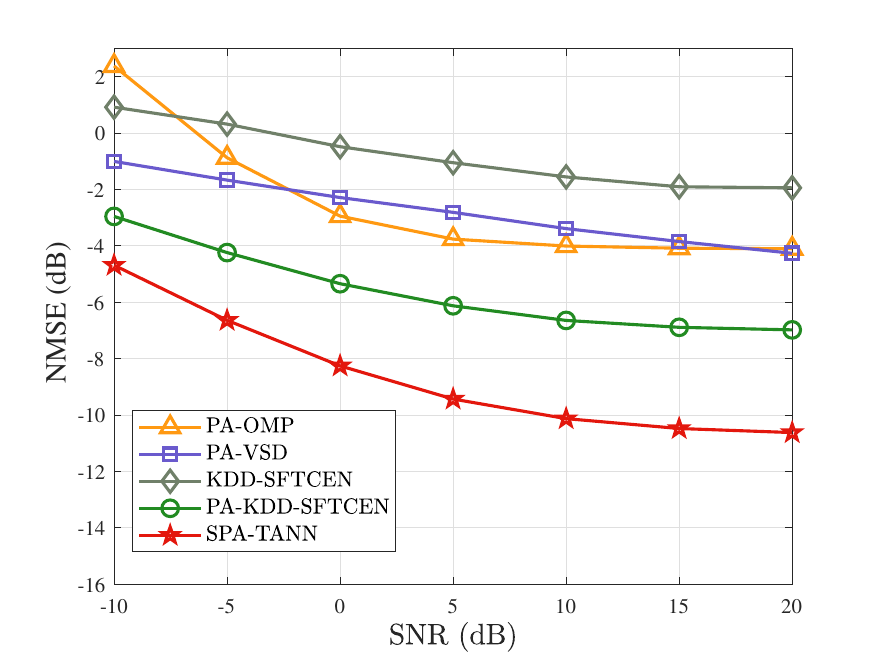}
            \label{fig:sim_NMSE_SNR_2s4f}
        }

        \caption{Channel extrapolation performance versus SNR at $v=60$~km/h and $f_{\mathrm{c}}=15$~GHz.}
        \label{fig:sim_NMSE_SNR}
    \end{minipage}%
    \hspace{0.015\textwidth}%
    \begin{minipage}[t]{0.492\textwidth}
        \vspace{0pt}
        \centering

        \setcounter{subfigure}{0}

        \subfloat[Frequency domain, $N_{\mathrm{s}}=1$]{
            \includegraphics[width=0.47\linewidth]{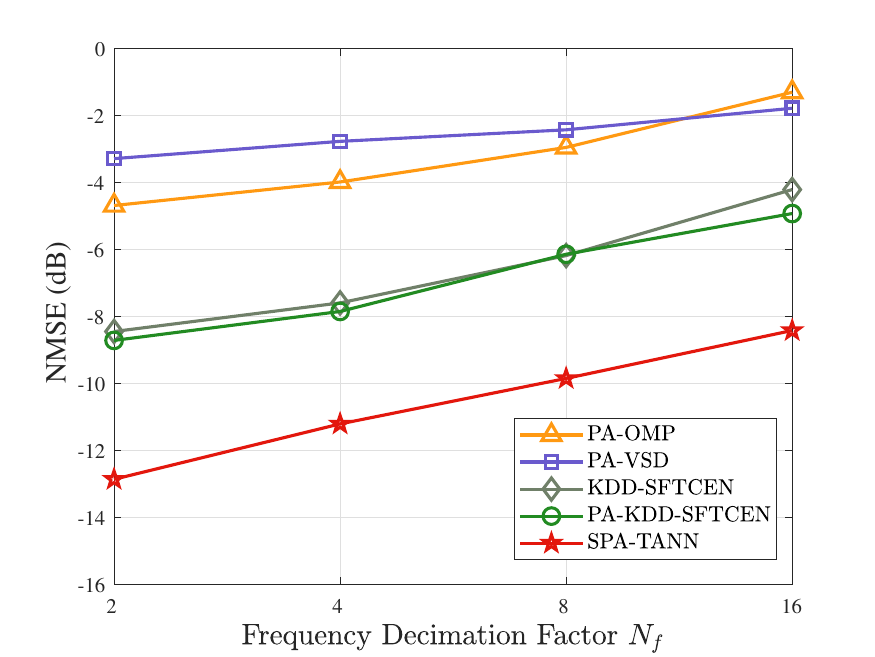}
            \label{Fig:sim_NMSE_DifferentNf_SNR0dB_60kmh}
        }%
        \hfill
        \subfloat[Spatial domain, $N_{\mathrm{f}}=2$]{
            \includegraphics[width=0.47\linewidth]{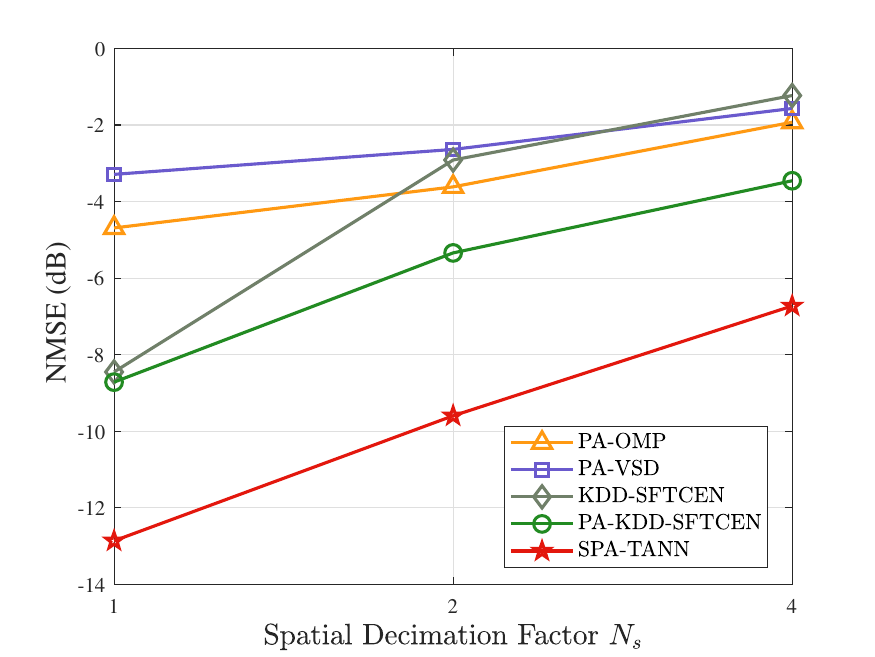}
            \label{Fig:sim_NMSE_DifferentNs_SNR0dB_60kmh}
        }

        \caption{Channel extrapolation performance versus decimation factors at $\text{SNR}=0$~dB, $v=60$~km/h and $f_{\mathrm{c}}=15$~GHz. }
        \label{fig:sim_NMSE_DecimationFactor}
    \end{minipage}%
    \vspace{-7mm}
\end{figure*}

\begin{figure*}[htbp]
    \centering
    \begin{minipage}[t]{0.492\textwidth}
        \vspace{0pt}
        \centering

        \subfloat[$N_{\mathrm{s}}=1$, $N_{\mathrm{f}}=2$]{
            \includegraphics[width=0.47\linewidth]{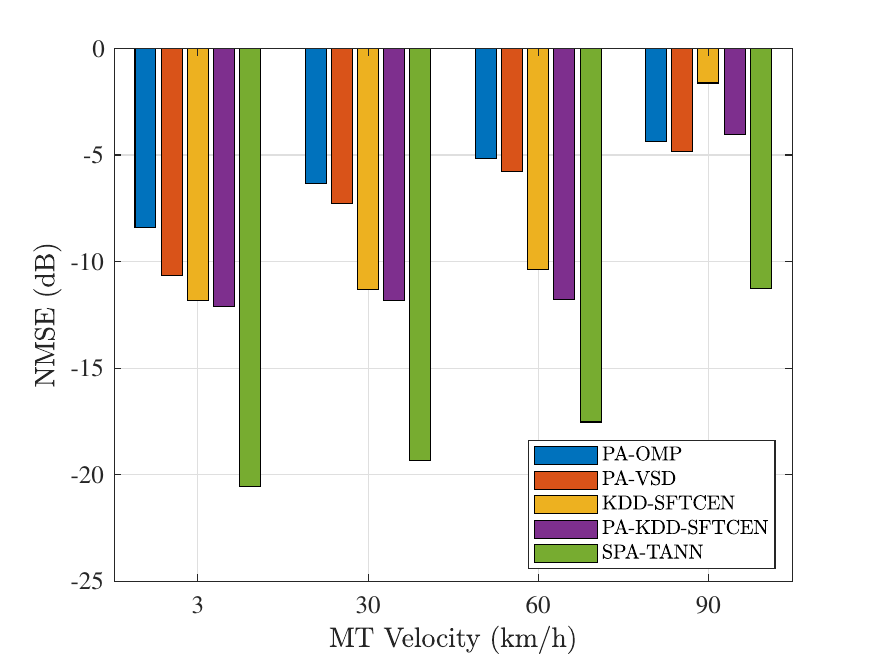}
            \label{fig:sim_NMSE_DifferentVelocity_1s2f}
        }%
        \hfill
        \subfloat[$N_{\mathrm{s}}=2$, $N_{\mathrm{f}}=4$]{
            \includegraphics[width=0.47\linewidth]{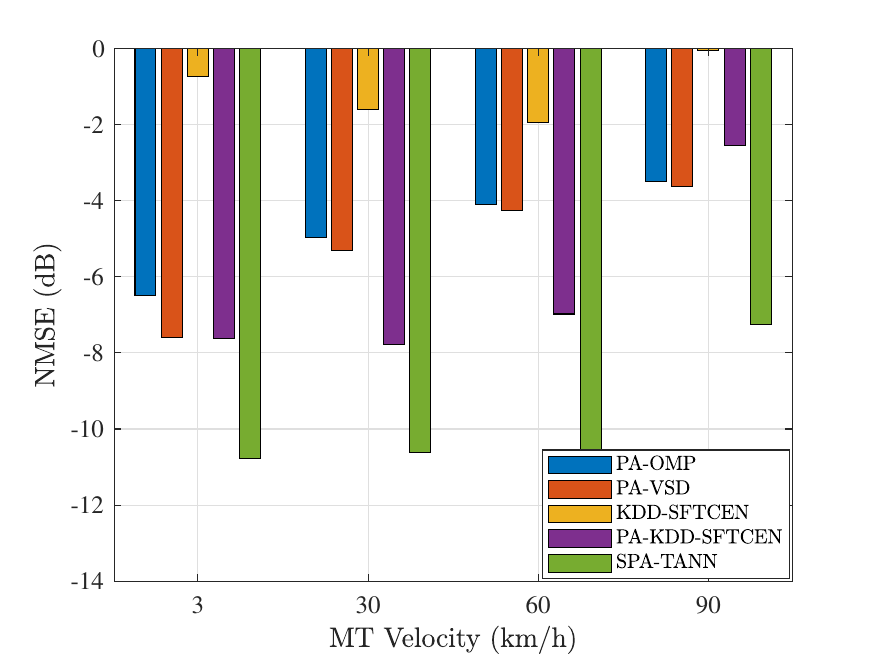}
            \label{fig:sim_NMSE_DifferentVelocity_2s4f}
        }

        \caption{Channel extrapolation performance versus MT velocity $v$ at $\text{SNR}=20$~dB and $f_{\mathrm{c}}=15$~GHz.}
        \label{fig:sim_NMSE_DifferentVelocity}
    \end{minipage}%
     \hspace{0.015\textwidth}%
    \begin{minipage}[t]{0.492\textwidth}
        \vspace{0pt}
        \centering

        \setcounter{subfigure}{0}

        \subfloat[$N_{\mathrm{s}}=1$, $N_{\mathrm{f}}=2$]{
            \includegraphics[width=0.47\linewidth]{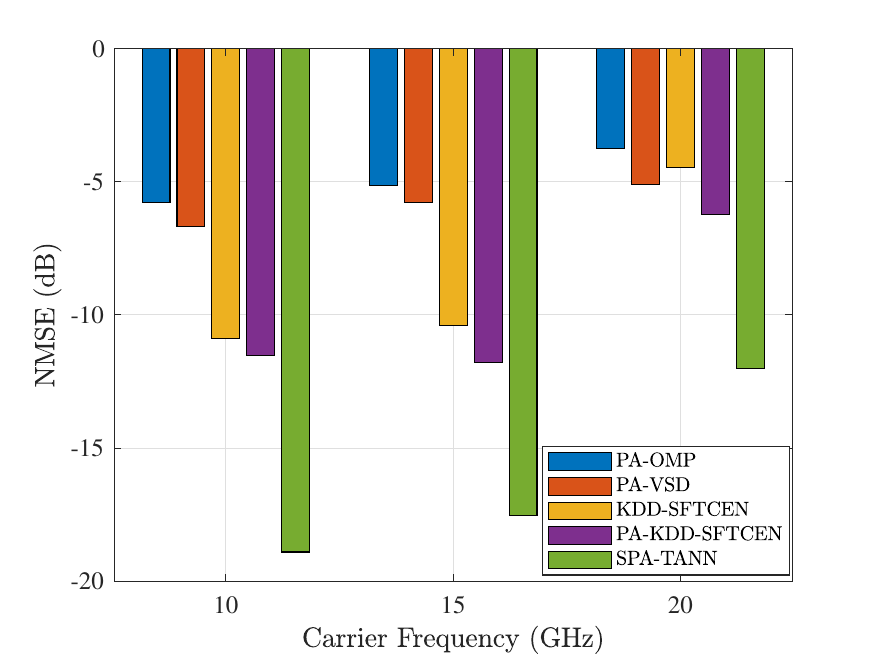}
            \label{fig:sim_NMSE_Differentfc_1s2f}
        }%
        \hfill
        \subfloat[$N_{\mathrm{s}}=2$, $N_{\mathrm{f}}=4$]{
            \includegraphics[width=0.47\linewidth]{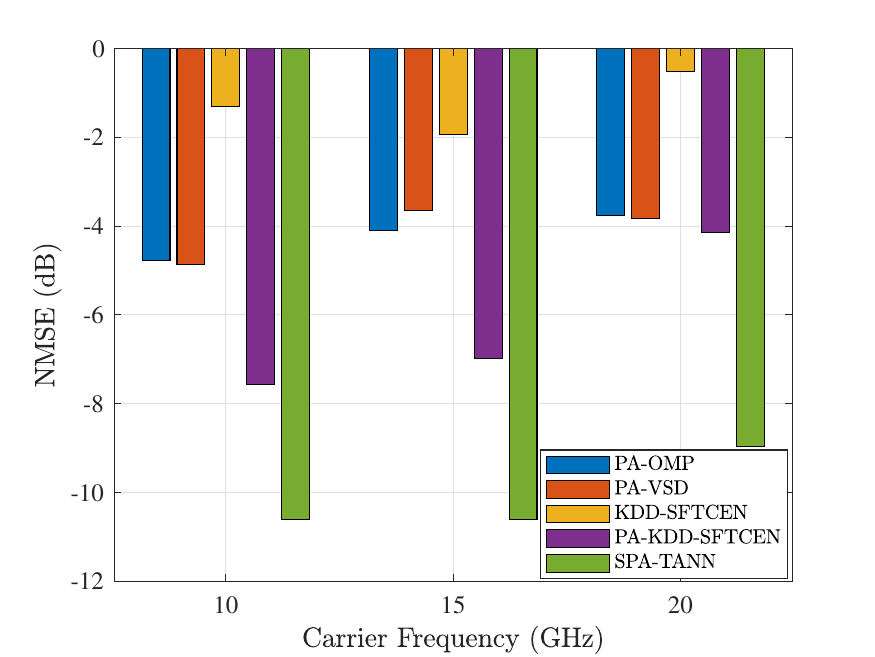}
            \label{fig:sim_NMSE_Differentfc_2s4f}
        }

        \caption{Channel extrapolation performance versus carrier frequency $f_{\mathrm{c}}$ at $\text{SNR}=20$~dB and $v=60$~km/h.}
        \label{fig:sim_NMSE_Differentfc}
    \end{minipage}%
    \vspace{-6mm}
\end{figure*}

\subsubsection{NMSE Versus Pilot-Decimation Factors} 
Fig.~\ref{fig:sim_NMSE_DecimationFactor} presents the NMSE with respect to
the frequency and spatial decimation factors \(N_{\mathrm{f}}\) and \(N_{\mathrm{s}}\),
respectively.
In Fig.~\ref{Fig:sim_NMSE_DifferentNf_SNR0dB_60kmh} with \(N_{\mathrm{s}}=1\),
increasing \(N_{\mathrm{f}}\) reduces the pilot density and intensifies delay-domain
aliasing, causing NMSE degradation for all methods. The proposed SPA-TANN
consistently achieves the lowest NMSE over the entire \(N_{\mathrm{f}}\) range, and
its performance at \(N_{\mathrm{f}}=16\) is comparable to that of KDD-SFTCEN at
\(N_{\mathrm{f}}=2\), suggesting up to an \(8\)-fold frequency-domain pilot overhead
reduction with comparable extrapolation accuracy. In Fig.~\ref{Fig:sim_NMSE_DifferentNs_SNR0dB_60kmh} with \(N_{\mathrm{f}}=2\),
increasing \(N_{\mathrm{s}}\) aggravates angle-domain aliasing and degrades all
schemes. The performance gap between KDD-SFTCEN and PA-KDD-SFTCEN indicates that
angle-domain priors are essential for alleviating the ambiguity introduced by
spatial pilot decimation.
 The proposed SPA-TANN achieves the best NMSE throughout, with
its performance at \(N_{\mathrm{s}}=4\) comparable to KDD-SFTCEN at \(N_{\mathrm{s}}=1\),
indicating up to a \(4\)-fold spatial-domain pilot reduction. These
results demonstrate that the proposed method enables aggressive pilot
decimation in both frequency and spatial domains while limiting the
associated performance degradation.

\subsubsection{NMSE Versus MT Velocity} 
We further examine the robustness of SPA-TANN to the MT velocity. The NMSE performance versus velocity is shown in Fig.~\ref{fig:sim_NMSE_DifferentVelocity} under two decimation configurations, where all data-driven methods are trained at $v=60$~km/h. 
As the MT velocity increases, the NMSE generally tends to increase due to the stronger temporal variation of the channel. For testing velocities below the training velocity, i.e., $v=3$ and $30$~km/h, the deep learning-based methods remain effective overall, since slower temporal channel variation usually corresponds to a less challenging extrapolation task. When the velocity increases to $v=90$~km/h beyond the training condition, the performance degradation becomes more pronounced for the baselines. In contrast, the proposed SPA-TANN consistently achieves the lowest NMSE  and exhibits a smaller performance degradation under velocity mismatch, demonstrating strong robustness enabled by the Doppler-domain modeling in the ADD-domain formulation.

\subsubsection{NMSE Versus Carrier Frequency}
Fig.~\ref{fig:sim_NMSE_Differentfc} evaluates all considered schemes at three carrier frequencies within the FR3 band, with all data-driven methods trained exclusively at $f_{\mathrm{c}}=15$~GHz. Since the Doppler frequency scales linearly with the carrier frequency, higher $f_{\mathrm{c}}$ leads to faster temporal channel variation, and accordingly increases the NMSE of all schemes in both cases. The proposed SPA-TANN consistently achieves the lowest NMSE, demonstrating its effectiveness throughout the FR3 band. Notably, the performance degradation of KDD-SFTCEN is more severe from $f_{\mathrm{c}}=15$~GHz to $20$~GHz than from $15$~GHz to $10$~GHz. This asymmetry arises because $f_{\mathrm{c}}=10$~GHz entails a moderate Doppler spread whose temporal characteristics remain well represented by the training data at $15$~GHz, whereas $f_{\mathrm{c}}=20$~GHz introduces faster temporal variation beyond the training distribution. In contrast, SPA-TANN suffers less degradation from $15$~GHz to $20$~GHz, benefiting from its physically grounded Doppler-domain representation that generalizes temporal channel behavior across carrier frequencies.

\subsubsection{NMSE for Different Scenarios}
\begin{table*}[!t]
    \centering
    \caption{NMSE (dB) of Channel Extrapolation Versus Scenarios at $\text{SNR}=20$~dB and $v=60$~km/h.}

    \label{tab:scenario}
    \scalebox{0.95}{\begin{tabular}{c c c c c c c}
    \toprule
    \multirow{2}{*}{\textbf{Pilot-Decimation Setting}} &
    \multirow{2}{*}{\textbf{Scenario}} &
    \multicolumn{5}{c}{\textbf{Algorithm}} \\
    \cmidrule(lr){3-7}
    & & SPA-TANN & PA-OMP & PA-VSD & KDD-SFTCEN & PA-KDD-SFTCEN \\
    \midrule
    \multirow{3}{*}{$(N_{\mathrm{s}},N_{\mathrm{f}})=(1,2)$}
    &  UMa NLoS & $\mathbf{-16.87}$ & $-3.87$ & $-3.71$ & $-8.56$ & $-9.77$ \\
    &  UMi NLoS & $\mathbf{-17.26}$ & $-2.91$ & $-1.84$ & $-10.22$ & $-11.50$ \\
    & RMa NLoS & $\mathbf{-18.21}$ & $-5.16$ & $-5.78$ & $-10.73$ & $-11.85$ \\
    \midrule
    \multirow{3}{*}{$(N_{\mathrm{s}},N_{\mathrm{f}})=(2,4)$}
    & UMa NLoS & $\mathbf{-9.05}$ & $-2.17$ & $-1.57$ & $-2.73$ & $-4.19$ \\
    & UMi NLoS & $\mathbf{-8.52}$ & $-1.95$ & $-0.75$ & $-3.03$ & $-5.19$ \\
    & RMa NLoS & $\mathbf{-13.29}$ & $-4.10$ & $-4.26$ & $-2.24$ & $-7.57$ \\
    \bottomrule
    \end{tabular}
    }
    \vspace{-8mm}
\end{table*}

\begin{figure}[!t] \centering 
    \includegraphics[width=0.3\textwidth]{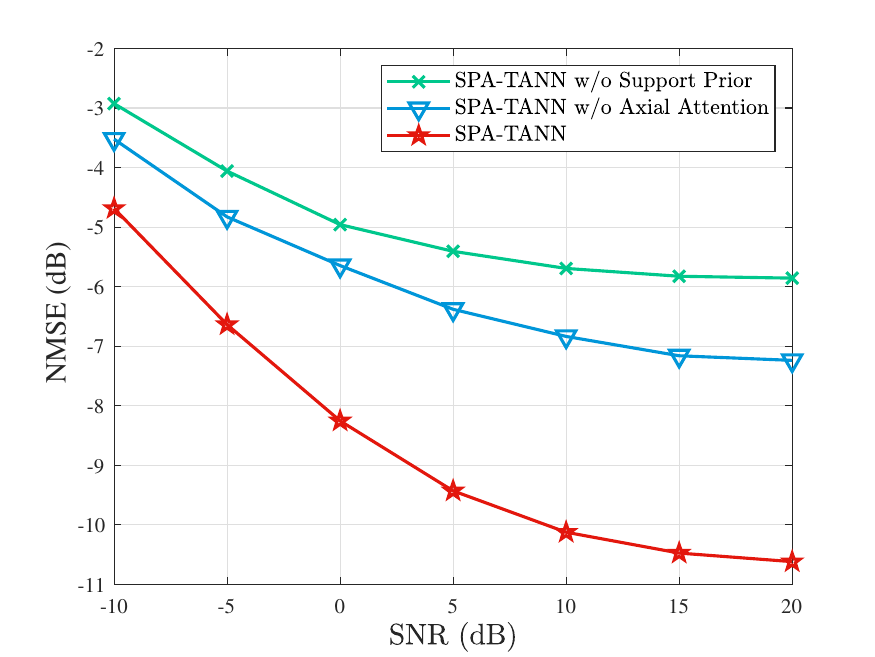} 
\caption{NMSE versus SNR for SPA-TANN and its ablated variants.}
    \label{fig:ablation_nmse_snr_2s4f} \vspace{-6mm} 
\end{figure}
\begin{figure*}[h]
    \centering

    \subfloat[Alias-Group Input]{
        \includegraphics[width=0.23\linewidth]{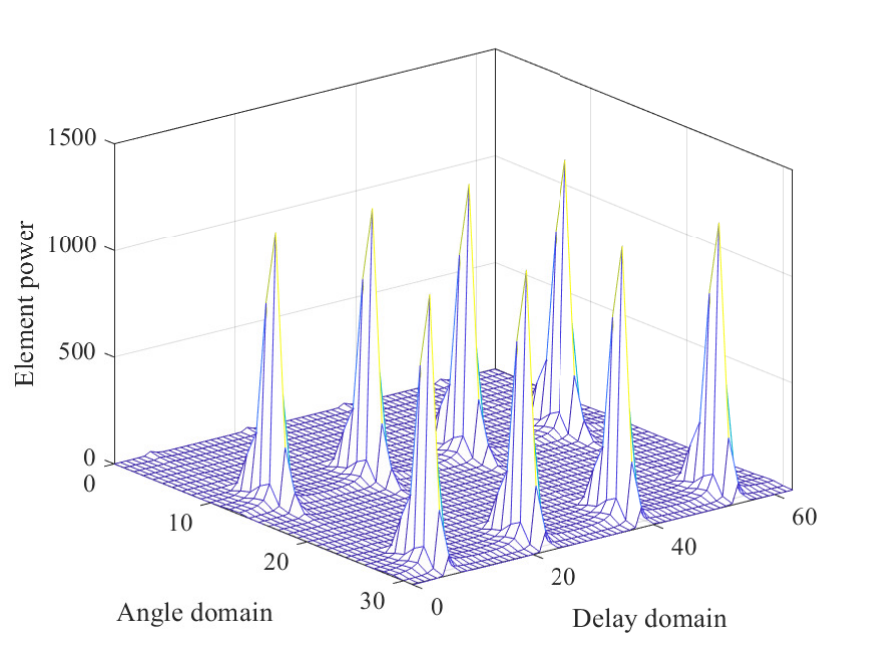}
        \label{fig:sim_alias_group}
    }
    \hfill
    \subfloat[Ground Truth]{
        \includegraphics[width=0.23\linewidth]{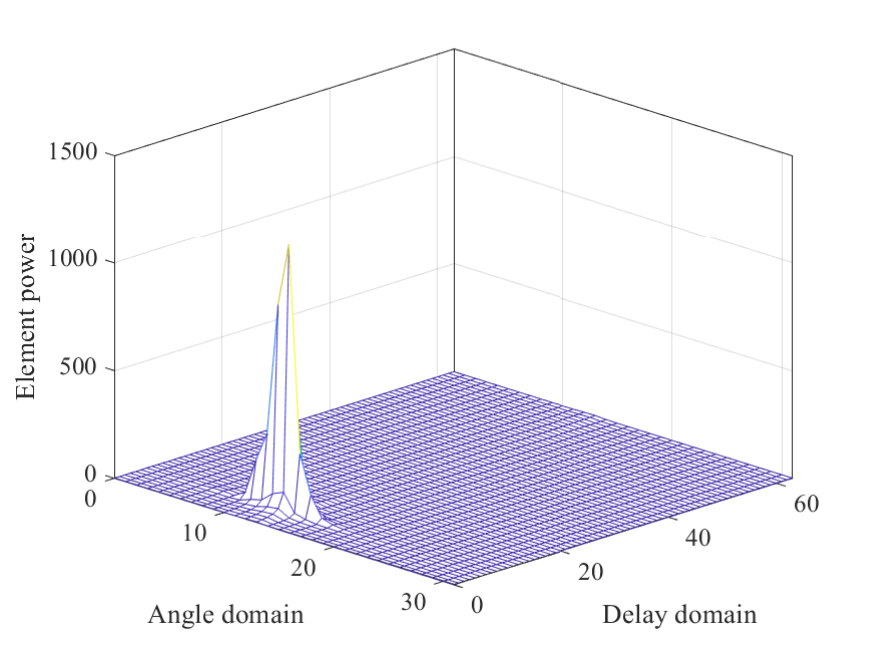}
        \label{fig:sim_ground_truth}
    }
    \hfill
    \subfloat[SPA-TANN Output]{
        \includegraphics[width=0.23\linewidth]{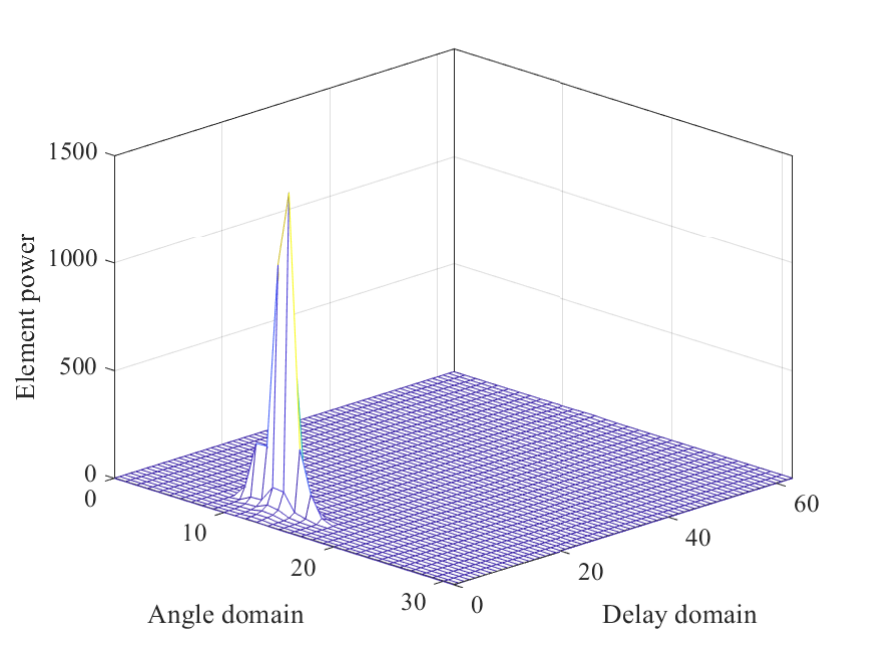}
        \label{fig:sim_TANN_output}
    }
    \hfill
    \subfloat[w/o Support Prior]{
        \includegraphics[width=0.23\linewidth]{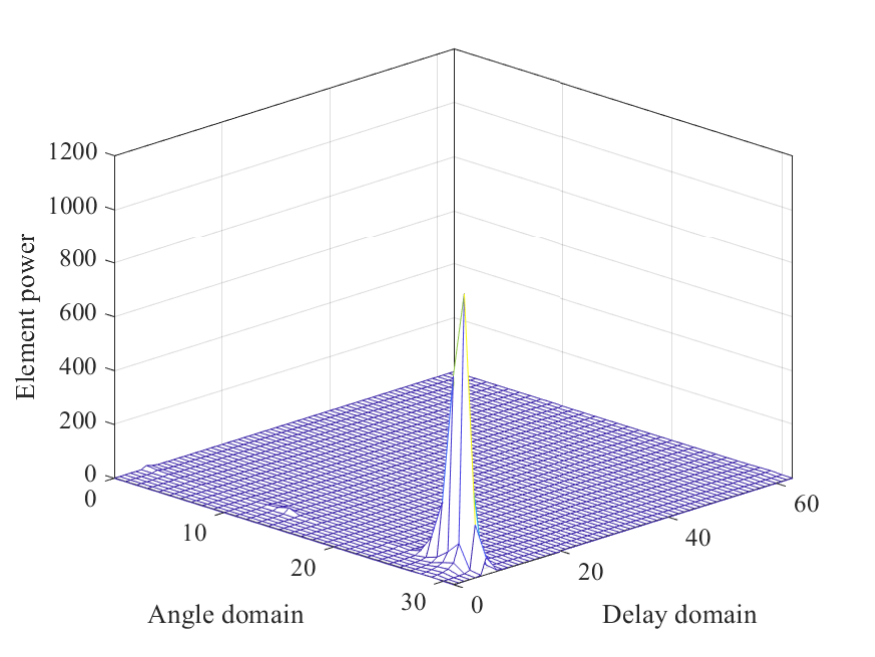}
        \label{fig:sim_withoutPrior}
    }
    \captionsetup{justification=raggedright,singlelinecheck=false}
    \caption{Angle-delay power spectrum visualization under
    \((N_{\mathrm{s}},N_{\mathrm{f}}) = (2,4)\), \(\text{SNR}=20\)~dB, \(v=60\)~km/h,
    and \(f_{\mathrm{c}}=15\)~GHz.}
    \label{fig:add_power_spectrum}
    \vspace{-7mm}
\end{figure*}

To demonstrate the robustness of the proposed method, we consider three
typical outdoor  scenarios defined in the 3GPP
specifications~\cite{3GPP_38901}: UMa NLoS, urban micro
(UMi) NLoS, and rural macro (RMa) NLoS. UMa NLoS represents macro-cellular
urban propagation with moderate multipath, UMi NLoS exhibits richer local
scattering with larger delay and angle spreads, and RMa NLoS is
characterized by sparser scattering and smaller delay-angle spreads.
All deep learning-based methods are trained on a mixed-scenario dataset containing the same number of samples as single-scenario training, and their NMSE performance under two pilot-decimation settings is summarized in Table~\ref{tab:scenario}.
For $(N_{\mathrm{s}},N_{\mathrm{f}})=(1,2)$, the performance variation across propagation scenarios remains limited, while SPA-TANN consistently outperforms all baselines. 
As the decimation factors increase to $(N_{\mathrm{s}},N_{\mathrm{f}})=(2,4)$, the scenario-dependent performance gap becomes more pronounced. The sparse channel structure in RMa NLoS facilitates extrapolation, whereas the richer multipath and larger delay-angle spreads in urban NLoS scenarios exacerbate aliasing ambiguity. 
Across both settings, PA-KDD-SFTCEN improves upon KDD-SFTCEN through PA-based preprocessing, while SPA-TANN achieves the lowest NMSE in all evaluated scenarios, demonstrating robust generalization under mixed-scenario training.

\vspace{-4mm}
\subsection{Ablation Study}
To validate the effectiveness of the key components of TANN, 
we conduct ablation experiments under the challenging decimation configuration 
$(N_{\mathrm{s}}, N_{\mathrm{f}})=(2,4)$ at $v=60$~km/h and 
$f_{\mathrm{c}}=15$~GHz, comparing the full model with two variants:
\begin{itemize}
	\item \textbf{SPA-TANN w/o Support Prior}: This variant removes the
	multi-scale CNN context gating module entirely, so that no support prior is injected and the network relies
	solely on the axial-attention mechanism to resolve aliasing
	from the LS-initialized ADD-domain tensor.
	\item \textbf{SPA-TANN w/o Axial-Attention}: This variant replaces
	the axis-wise Transformer-based attention blocks with 3D
	convolutional layers,
	evaluating the contribution of the global dependency modeling and feature representation capacity
	provided by axial-attention.
\end{itemize}
As shown in Fig.~\ref{fig:ablation_nmse_snr_2s4f}, both ablated
variants exhibit substantial NMSE degradation relative to the
 SPA-TANN across the entire SNR range, confirming that each
component contributes meaningfully to the overall performance.  The variant
without the support prior exhibits the most pronounced loss,
which confirms the importance of support prior-assisted context modeling
in mitigating angle-domain ambiguity caused by spatial
decimation. Replacing axial-attention with 3D CNN also degrades
the reconstruction accuracy, indicating that local convolutional
operations are less effective in capturing the long-range
aliasing dependencies across multiple domains. The
SPA-TANN, which integrates both components, consistently achieves
the lowest NMSE, verifying the complementary and indispensable
roles of the axial-attention architecture and the
SPA context gating in the proposed design.

To provide an intuitive understanding of how the support prior assists in resolving multi-domain aliasing, we visualize the angle-delay power spectrum of one representative test sample in Fig.~\ref{fig:add_power_spectrum}. Owing to the limited delay spread of wireless channels, the dominant components are confined to a compact low-delay support; this inherent concentration alleviates the ambiguity from frequency-domain decimation and renders the principal delay structure recoverable. The angle domain, lacking a comparable concentration property, is far more susceptible to spatial-decimation-induced ambiguity. Absent support-prior guidance, the ablated network may resolve the aliased angular components incorrectly, yielding noticeable deviations of the angular peaks from the ground truth. 
This limitation is effectively overcome by the SPA de-aliasing module in SPA-TANN, which suppresses the aliasing ambiguity and recovers a spectrum closely matching the ground truth.

\section{Conclusion}
\label{sec_Conclusion}
This paper investigated  multi-domain channel extrapolation for upper mid-band massive MIMO-OFDM systems. Through the Tucker-based ADD-domain analysis, we found that uniform  pilot decimation across the SFT domains   induces asymmetric aliasing across the angle, delay, and Doppler domains, which fundamentally limits channel extrapolation and necessitates support-prior guidance to resolve the ADD-domain aliasing. Guided by this finding, SPA-TANN combines axial attention and lightweight multi-scale gating for support-prior-assisted ADD-domain de-aliasing, and is trained over different pilot decimation factors to achieve cross-configuration generalization without retraining.  
Simulation results show that the proposed framework consistently achieves the lowest NMSE across diverse velocities, carrier frequencies, and scenarios, while reducing the frequency- and spatial-domain pilot overhead by up to eight-fold and four-fold, respectively, demonstrating its potential for efficient CSI acquisition in FR3 massive MIMO systems with limited pilot overhead.

\normalem

\bibliographystyle{IEEEtran}
\bibliography{SFTChannelExtrapolation}

@ARTICLE{3DDomainExtrapolationGaoFeiFei,
  author={Zhou, Binggui and Yang, Xi and Ma, Shaodan and Gao, Feifei and Yang, Guanghua},
  journal={IEEE Trans. Wireless Commun.}, 
  title={Low-Overhead Channel Estimation via {3D} Extrapolation for {TDD} {mmWave} Massive {MIMO} Systems Under High-Mobility Scenarios}, 
  year={2025},
  volume={24},
  number={4},
  pages={2797-2813},
  doi={10.1109/TWC.2024.3524911}}

@article{GaoFeiFei2025AntennaSelection,
  title={Deep learning-based channel extrapolation for {5G} advanced massive {MIMO}: Hardware prototype and experimental evaluation},
  author={Li, Hongyao and Wang, Mingjin and Han, Runyu and Wang, Ning and Wu, Huihui and Gu, Yuantao and Yuan, Wanmai and Gao, Feifei},
  journal={IEEE Trans. Wireless Commun.}, 
  volume={24},
  number={3},
  pages={1756--1771},
  year={2025},
  publisher={IEEE}
}

@article{11049893,
  author    = {Wang, Yafei and Hou, Hongwei and Yi, Xinping and Wang, Wenjin and Jin, Shi},
  title     = {Toward Unified {AI} Models for {MU-MIMO} Communications: A Tensor Equivariance Framework},
  journal   = {IEEE Trans. Wireless Commun.},
  year      = {2025},
  volume    = {24},
  number    = {12},
  pages     = {10517--10533},
  month     = {Dec.},
  doi       = {10.1109/TWC.2025.3580321},
  publisher = {IEEE},
  issn      = {1536-1276}
}

@ARTICLE{HouTuckerBasedChannelPrediction,
  author={Hou, Hongwei and Wang, Yafei and Zhu, Yiming and Yi, Xinping and Wang, Wenjin and Slock, Dirk T. M. and Jin, Shi},
	journal={IEEE Trans. Wireless Commun.}, 
  title={A Tensor-Structured Approach to Dynamic Channel Prediction for Massive {MIMO} Systems With Temporal Non-Stationarity}, 
  year={2026},
  volume={25},
  number={},
  pages={6869-6886},
}

@article{Wang2025DP,
  author  = {Wang, Yafei and Ha, Vu Nguyen and Ntontin, Konstantinos and Yan, Hong and Wang, Wenjin and Chatzinotas, Symeon and Ottersten, Bj{\"o}rn},
  title   = {Statistical {CSI}-Based Distributed Precoding Design for {OFDM}-Cooperative Multi-Satellite Systems},
  journal = {IEEE J. Sel. Areas Commun.},
 year={2026},
  volume={44},
  number={},
  pages={3219-3236},
}

@ARTICLE{XLMIMOCPHou,
  author={Hou, Hongwei and Wang, Yafei and Yi, Xinping and Wang, Wenjin and Slock, Dirk T. M. and Jin, Shi},
  journal={IEEE Trans. Commun.}, 
  title={Tensor-Structured {Bayesian} Channel Prediction for Upper Mid-Band {XL}-{MIMO} Systems}, 
  year={2026},
  volume={74},
  number={},
  pages={3968-3983},
  doi={10.1109/TCOMM.2026.3658933}}

@article{Cite2014Quadriga,
	title={{QuaDRiGa}: A 3-{D} multi-cell channel model with time evolution for enabling virtual field trials},
	author={Jaeckel, Stephan and Raschkowski, Leszek and B{\"o}rner, Kai and Thiele, Lars},
	journal={IEEE Trans. Antennas Propag.},
	volume={62},
	number={6},
	pages={3242--3256},
	month=jun,	
	year={2014},
	publisher={IEEE}
}

@book{3GPP_38901,
		title={Study on channel model for frequencies from 0.5 to 100 {GHz, Version} 16.1.0},
		year={2019},
		month=dec,
		publisher={document 3GPP T.R. 38.901}
	}

@article{CiteZengyongCKMHybridBeamforming,
	author={Wu, Di and Zeng, Yong and Jin, Shi and Zhang, Rui},
	journal={IEEE Trans. Wireless Commun.}, 
	title={Environment-Aware Hybrid Beamforming by Leveraging Channel Knowledge Map}, 
	volume={23},
	number={5},
	pages={4990-5005},
	month={May.},
	year={2024},
	publisher={IEEE}
}

@article{Citezeng2024CKMtutorialSCSI,
  author={Zeng, Yong and Chen, Junting and Xu, Jie and Wu, Di and Xu, Xiaoli and Jin, Shi and Gao, Xiqi and Gesbert, David and Cui, Shuguang and Zhang, Rui},
journal={IEEE Commun. Surveys Tuts.}, 
title={A Tutorial on Environment-Aware Communications via Channel Knowledge Map for {6G}}, 
month=feb,
year={2024},
volume={26},
number={3},
pages={1478-1519},
	publisher={IEEE}
}

@article{Enabling6GThroughMultiDomainChannelExtrapolation,
  title={Enabling {6G} through multi-domain channel extrapolation: Opportunities and challenges of generative artificial intelligence},
  author={Gao, Yuan and Lu, Zichen and Wu, Yifan and Jin, Yanliang and Zhang, Shunqing and Chu, Xiaoli and Xu, Shugong and Wang, Cheng-Xiang},
  journal={IEEE Commun. Mag.},
  year={2026},
  volume={64},
  number={1},
  pages={222-228},
  publisher={IEEE}
}

@article{Lin2021AntennaSelection,
  title={Deep learning-based antenna selection and {CSI} extrapolation in massive {MIMO} systems},
  author={Lin, Bo and Gao, Feifei and Zhang, Shun and Zhou, Ting and Alkhateeb, Ahmed},
  journal={IEEE Trans. Wireless Commun.},
  volume={20},
  number={11},
  pages={7669--7681},
  year={2021},
  publisher={IEEE}
}

@inproceedings{yang2020deep,
  title={Deep learning based antenna selection for channel extrapolation in {FDD} massive {MIMO}},
  author={Yang, Yindi and Zhang, Shun and Gao, FeiFei and Xu, Chao and Ma, Jianpeng and Dobre, Octavia A},
  booktitle={Proc. IEEE Int. Conf. Wireless Commun. Signal Process. (WCSP)},
  pages={182--187},
  year={2020},
  organization={IEEE}
}

@article{He2025JUDCEN,
  title={Deep Learning-Based Joint Uplink-Downlink {CSI} Acquisition for Next-Generation Upper Mid-Band Systems},
  author={He, Xuan and Hou, Hongwei and Wang, Yafei and Wang, Wenjin and Jin, Shi and Chatzinotas, Symeon and Ottersten, Bj{\"o}rn},
  journal={arXiv preprint arXiv:2512.02557},
  year={2025}
}

@article{Gao2026PDP,
  title={Channel Extrapolation for {MIMO} Systems with the Assistance of Multi-path Information Induced from Channel State Information},
  author={Gao, Yuan and Wu, Xinyi and Jun, Jiang and Zhang, Zitian and Yang, Zhaohui and Xu, Shugong and Wang, Cheng-Xiang and Han, Zhu},
  journal={arXiv preprint arXiv:2601.21524},
  year={2026}
}

@ARTICLE{XRLoopbackMechanismBojovic2023,
  author={Bojović, Biljana and Lagén, Sandra and Koutlia, Katerina and Zhang, Xiaodi and Wang, Ping and Yu, Liwen},
  journal={IEEE J. Sel. Areas Commun.}, 
  title={Enhancing {5G} {QoS} Management for {XR} Traffic Through {XR} Loopback Mechanism}, 
  year={2023},
  volume={41},
  number={6},
  pages={1772-1786},
  doi={10.1109/JSAC.2023.3273701}}

@article{BroadbandIoTZhou2023,
  title={Novel listen-before-talk access scheme with adaptive backoff procedure for uplink centric broadband communication},
  author={Zhou, Hui and Deng, Yansha and Nallanathan, Arumugam},
  journal={IEEE Internet Things J.},
  volume={10},
  number={22},
  pages={19981--19992},
  year={2023},
  publisher={IEEE}
}

@article{GaofeifeiDACEN2023,
  title={Pay less but get more: A dual-attention-based channel estimation network for massive {MIMO} systems with low-density pilots},
  author={Zhou, Binggui and Yang, Xi and Ma, Shaodan and Gao, Feifei and Yang, Guanghua},
  journal={IEEE Trans. Wireless Commun.},
  year={2024},
  volume={23},
  number={6},
  pages={6061-6076},
  publisher={IEEE}
}

@article{MMSEOFDMChannelEstimationEdfors1998,
  title={{OFDM} channel estimation by singular value decomposition},
  author={Edfors, Ove and Sandell, Magnus and Van de Beek, J-J and Wilson, Sarah Kate and Borjesson, P Ola},
  journal={IEEE Trans. Commun.},
  volume={46},
  number={7},
  pages={931--939},
  year={1998},
  publisher={IEEE}
}

@article{OFDMLinearInterpolationLee2008,
  title={On the training of {MIMO}-{OFDM} channels with least square channel estimation and linear interpolation},
  author={Lee, Seung Joon},
  journal={IEEE Commun. Lett.},
  volume={12},
  number={2},
  pages={100--102},
  year={2008},
  publisher={IEEE}
}

@article{2DChannelExtrapolationWan2024,
  title={A two-stage {2D} channel extrapolation scheme for {TDD} {5G} {NR} systems},
  author={Wan, Yubo and Liu, An},
  journal={IEEE Trans. Wireless Commun.},
  volume={23},
  number={8},
  pages={8497--8511},
  year={2024},
  publisher={IEEE}
}

@ARTICLE{WuchiARChannelPrediction,
  author={Wu, Chi and Yi, Xinping and Zhu, Yiming and Wang, Wenjin and You, Li and Gao, Xiqi},
  journal={IEEE J. Sel. Areas Commun.}, 
  title={Channel Prediction in High-Mobility Massive {MIMO}: From Spatio-Temporal Autoregression to Deep Learning}, 
  year={2021},
  volume={39},
  number={7},
  pages={1915-1930},
  doi={10.1109/JSAC.2021.3078503}}

@ARTICLE{Yuan2020MachineLearningBasedChannelPredictionAR,
  author={Yuan, Jide and Ngo, Hien Quoc and Matthaiou, Michail},
  journal={IEEE Trans. Wireless Commun.},
  title={Machine Learning-Based Channel Prediction in Massive {MIMO} With Channel Aging}, 
  year={2020},
  volume={19},
  number={5},
  pages={2960-2973},
  doi={10.1109/TWC.2020.2969627}}

@inproceedings{SosChannelPredictionWong2006,
  title={{WLC43-5}: low-complexity adaptive high-resolution channel prediction for {OFDM} systems},
  author={Wong, Ian C and Evans, Brian L},
  booktitle={Proc. IEEE Global Telecommun. Conf. (GLOBECOM)},
  pages={1--5},
  year={2006},
  organization={IEEE}
}

@article{ParameterBasedModel2018prediction,
  title={Prediction of time-varying multi-user {MIMO} channels based on {DOA} estimation using compressed sensing},
  author={Uehashi, Shunsuke and Ogawa, Yasutaka and Nishimura, Toshihiko and Ohgane, Takeo},
  journal={IEEE Trans. Veh. Technol.},
  volume={68},
  number={1},
  pages={565--577},
  year={2019},
  publisher={IEEE}
}

@article{LSTMGRUChannelPredictionHelmy2023,
  title={{LSTM}-{GRU} model-based channel prediction for one-bit massive {MIMO} system},
  author={Helmy, Islam and Tarafder, Pulok and Choi, Wooyeol},
  journal={IEEE Trans. Veh. Technol.},
  volume={72},
  number={8},
  pages={11053--11057},
  year={2023},
  publisher={IEEE}
}

@ARTICLE{TrasnformerChannelPredictionJiang2022,
  author={Jiang, Hao and Cui, Mingyao and Ng, Derrick Wing Kwan and Dai, Linglong},
  journal={IEEE J. Sel. Areas Commun.}, 
  title={Accurate Channel Prediction Based on Transformer: Making Mobility Negligible}, 
  year={2022},
  volume={40},
  number={9},
  pages={2717-2732},
  doi={10.1109/JSAC.2022.3191334}}

@ARTICLE{MultiDomainChannelExtrapolationHan2021,
  author={Han, Yu and Jin, Shi and Li, Xiao and Wen, Chao-Kai and Quek, Tony Q. S.},
  journal={IEEE Trans. Commun.}, 
  title={Multi-Domain Channel Extrapolation for {FDD} Massive {MIMO} Systems}, 
  year={2021},
  volume={69},
  number={12},
  pages={8534-8550},
  doi={10.1109/TCOMM.2021.3112593}}

@article{CiteOMPBaseLine,
		title={Signal recovery from random measurements via orthogonal matching pursuit},
		author={Tropp, Joel A and Gilbert, Anna C},
		journal={IEEE Trans. Inf. Theory},
		volume={53},
		number={12},
		pages={4655--4666},
		year={2007},
				month=dec,
		publisher={IEEE}
	}

@article{CiteVSDTensorDecompositionMethod,
		title={Estimating Channels With Hundreds of Sub-Paths for {MU-MIMO} Uplink: A Structured High-Rank Tensor Approach},
		author={Chen, Panqi and Cheng, Lei},
		journal= {IEEE Signal Process. Lett.},
		year={2024},
		month= sep,
		volume={31},
		pages={2320-2324},
		publisher={IEEE}
	}

@ARTICLE{HBFSurveyAhmed2018,
  author={Ahmed, Irfan and Khammari, Hedi and Shahid, Adnan and Musa, Ahmed and Kim, Kwang Soon and De Poorter, Eli and Moerman, Ingrid},
  journal={IEEE Commun. Surveys Tuts.}, 
  title={A Survey on Hybrid Beamforming Techniques in {5G}: Architecture and System Model Perspectives}, 
  year={2018},
  volume={20},
  number={4},
  pages={3060-3097},
  doi={10.1109/COMST.2018.2843719}}

@ARTICLE{FR3ChannelMeasurementMiao2023,
  author={Miao, Haiyang and Zhang, Jianhua and Tang, Pan and Tian, Lei and Zhao, Xinyu and Guo, Bolun and Liu, Guangyi},
  journal={IEEE J. Sel. Areas Commun.}, 
  title={Sub-6 {GHz} to {mmWave} for {5G}-Advanced and Beyond: Channel Measurements, Characteristics and Impact on System Performance}, 
  year={2023},
  volume={41},
  number={6},
  pages={1945-1960},
  doi={10.1109/JSAC.2023.3274175}}

@ARTICLE{HBFSwithchChannelEstimationMao2018,
  author={Vlachos, Evangelos and Alexandropoulos, George C. and Thompson, John},
  journal={IEEE Signal Process. Lett.}, 
  title={Massive {MIMO} Channel Estimation for Millimeter Wave Systems via Matrix Completion}, 
  year={2018},
  volume={25},
  number={11},
  pages={1675-1679},
  doi={10.1109/LSP.2018.2870533}}

@article{vaswani2017attention,
  title={Attention is all you need},
  author={Vaswani, Ashish and Shazeer, Noam and Parmar, Niki and Uszkoreit, Jakob and Jones, Llion and Gomez, Aidan N and Kaiser, {\L}ukasz and Polosukhin, Illia},
  journal={Advances in neural information processing systems},
  volume={30},
  year={2017}
}

@article{FR1FR2ChannelMeasurementNaqvi2021,
  title={5{G} {NR} {mmWave} indoor coverage with massive antenna system},
  author={Naqvi, Syed Hassan Raza and Ho, Pin Han and Peng, Limei},
  journal={J. Commun. Networks},
  volume={23},
  number={1},
  pages={1--11},
  year={2021},
  publisher={KICS}
}

@article{hou2024jointCIR,
  title={Joint beam alignment and {Doppler} estimation for fast time-varying wideband {mmWave} channels},
  author={Hou, Hongwei and Wang, Yafei and Yi, Xinping and Wang, Wenjin and Jin, Shi},
  journal={IEEE Trans. Wireless Commun.},
  volume={23},
  number={9},
  pages={5477--5492},
  year={2024},
  publisher={IEEE}
}

@ARTICLE{SpatialDomainChannelExtrapolationZhang2022,
  author={Zhang, Shunbo and Zhang, Shun and Ma, Jianpeng and Liu, Tian and Dobre, Octavia A.},
  journal={IEEE Trans. Veh. Technol.}, 
  title={Deep Learning Based Antenna-Time Domain Channel Extrapolation for Hybrid {mmWave} Massive {MIMO}}, 
  year={2022},
  volume={71},
  number={12},
  pages={13398-13402},
  doi={10.1109/TVT.2022.3197452}}

@ARTICLE{OFDM1,
  author={Zhu, Huiling and Wang, Jiangzhou},
  journal={IEEE Trans. Commun.}, 
  title={Chunk-based resource allocation in {OFDMA} systems - part {I}: chunk allocation}, 
  year={2009},
  volume={57},
  number={9},
  pages={2734-2744},
  }

@ARTICLE{OFDM2,
  author={Zhu, Huiling and Wang, Jiangzhou},
  journal={IEEE Trans. Commun.}, 
  title={Chunk-Based Resource Allocation in {OFDMA} Systems—Part {II}: Joint Chunk, Power and Bit Allocation}, 
  year={2012},
  volume={60},
  number={2},
  pages={499-509},
}

@ARTICLE{MiaoFR3Measurement,
  author={Miao, Haiyang and Zhang, Jianhua and Tang, Pan and Zhen, Qi and Meng, Jie and Liu, Ximan and Liu, Enrui and Liu, Peijie and Tian, Lei and Liu, Guangyi},
  journal={ IEEE Open J. Commun. Soc.}, 
  title={{6G} New Mid-Band/{FR3} (6–24 {GHz}): Channel Measurement, Characteristics and Modeling}, 
  year={2025},
  volume={6},
  number={},
  pages={9942-9960},
 }

@article{LiuFR3Measurement,
  title={Measurement-Based Analysis of Outdoor Massive {MIMO} Channel Characteristics over {FR3} Frequency Band},
  author={Liu, Enrui and Tang, Pan and Miao, Haiyang and Zhen, Qi and Zhang, Jianhua and Wang, Sen},
  journal={arXiv preprint arXiv:2606.11622},
  year={2026}
}

\end{document}